\documentclass[11pt]{article}
\usepackage{amsmath,amsthm,amssymb,mathtools}
\usepackage[a4paper,margin=1in]{geometry}
\usepackage{microtype}
\usepackage{hyperref}

\hypersetup{
  hidelinks,
  pdftitle={Polar Fidelities and the Geometry of Holevo Bases},
  pdfauthor={Trung Dung Vuong},
  pdfsubject={Polar monotonicity, data processing, and Holevo bases for generalized fidelity},
  pdfkeywords={generalized fidelity, polar fidelity, Holevo fidelity, positive factorization, unitary-factor fibers}
}

\newtheorem{theorem}{Theorem}[section]
\newtheorem{proposition}[theorem]{Proposition}
\newtheorem{corollary}[theorem]{Corollary}
\newtheorem{lemma}[theorem]{Lemma}
\theoremstyle{remark}

\newtheorem{example}[theorem]{Example}

\newcommand{\Pd}{\mathbb{P}_d}
\newcommand{\Tr}{\operatorname{Tr}}
\newcommand{\Pol}{\operatorname{Pol}}
\newcommand{\Herm}{\operatorname{Herm}}
\newcommand{\FM}{F^{\mathrm M}}
\newcommand{\FH}{F^{\mathrm H}}
\newcommand{\FU}{F^{\mathrm U}}
\newcommand{\FLE}{F^{\mathrm{LE}}}
\newcommand{\Fpol}[1]{F^{\mathrm{pol}}_{#1}}
\begin{document}
  
\title{Polar Fidelities and the Geometry of Holevo Bases}
\author{Trung Dung Vuong\textsuperscript{1,2}\\[3pt]
	\small\textsuperscript{1}High School for the Gifted, VNUHCM\\
	\small 153 Nguyen Chi Thanh, An Dong Ward, Ho Chi Minh City, Vietnam\\
	\small\textsuperscript{2}Vietnam National University Ho Chi Minh City\\
	\small Linh Xuan Ward, Ho Chi Minh City, Vietnam\\
	\small\texttt{vtdung@ptnk.edu.vn}}
\date{}
\maketitle
\begin{abstract}
We strengthen the polar monotonicity conjecture of Afham and Ferrie
by proving that, for noncommuting positive definite matrices $P,Q$,
the generalized fidelity along both paths $R=P^x$ and $R=Q^x$
increases strictly up to $x=1$ and decreases strictly thereafter.
The corresponding R\'enyi traces have the same two-branch monotonicity,
with explicitly determined maximizing exponents for $0<\alpha<1$.
Both one-sided fidelities and their average satisfy universal data
processing on positive definite inputs and outputs exactly when
$x\in[-1,0]\cup\{1\}$.
At noncommuting pairs, either polar path realizes the $z$-fidelity
precisely when $z\ge1/2$, including the Log-Euclidean limit.
The representing exponent and its normalized base admit global
continuous extensions at every finite order $z\ge1/2$ in dimension
two, but only at $z=1/2$ and $z=1$ in higher dimensions.

The attainable unitary factors and their fibers are completely
characterized, with the global range consisting of the identity
and all nonscalar special unitaries. A noncommuting pair admits
at most two Holevo factors in dimension two, whereas nonidentity
Holevo factors accumulate at the identity in higher dimensions.
In dimension three, a complete classification of Holevo factors
is obtained, and the simple-spectrum locus is dense in the full
factor set for noncommuting pairs.
\end{abstract}



\section{Introduction}
\label{sec:intro}

For positive definite complex $d\times d$ matrices $P,Q,R$, Afham and
Ferrie \cite[Definition~3.1]{AfhamFerrie2024} defined the generalized
fidelity with base $R$ by
\begin{equation}
 F_R(P,Q)=\Tr\!\left[(R^{1/2}PR^{1/2})^{1/2}R^{-1}
                        (R^{1/2}QR^{1/2})^{1/2}\right].
 \label{eq:generalized-fidelity}
\end{equation}
On either polar path $R=P^x$ or $R=Q^x$, the exponents $-1,0,1$
give the Matsumoto, Holevo and Uhlmann fidelities, respectively.
Although $F_R$ may be complex, its values along these paths are positive.
Afham and Ferrie asked whether each one-sided value is nondecreasing
on $[-1,1]$, which would also give monotonicity of the average
$\Fpol{x}(P,Q)=\tfrac12(F_{P^x}(P,Q)+F_{Q^x}(P,Q))$.
Their questions about representing other fidelities lead to two inverse
problems for \eqref{eq:generalized-fidelity}. Can one choose a base on
a polar path to realize a prescribed $z$-fidelity or the Log-Euclidean
fidelity? For the Holevo value, which bases in the entire positive cone
give that value? The polar monotonicity question is Open Problem~4 of
Afham and Ferrie \cite[Sec.~X]{AfhamFerrie2024}.
The two inverse problems above refine their Open Problems~3 and~5,
respectively.

Theorem~\ref{thm:one-sided-polar-monotone} establishes the
monotonicity of both polar paths on either side of their maximizing
parameters and characterizes the equality cases for every
R\'enyi order $0<\alpha<1$. In particular, at $\alpha=1/2$,
both paths are nondecreasing on $(-\infty,1]$ and nonincreasing
on $[1,\infty)$.  The response to a quantum channel depends
on how the base is chosen at the output. With the polar base formed
from the output pair, Theorem~\ref{thm:polar-dpi-classification}
classifies all $x\in\mathbb R$ for which fidelity is nondecreasing under
every channel.  On $[-1,0]$ the one-sided inequalities follow from data
processing for the geodesic R\'enyi functional of Capel and Costa Rico
\cite[Corollary~6.3]{CapelCostaRico2026}. Outside $[-1,0]\cup\{1\}$, counterexamples to data processing
are constructed using preparation channels or partial traces
with faithful input and output states.  For the related
generalized Bures-Wasserstein distance with transported base
$\Lambda(R)$, Ma, Rajaei and Zhang \cite{MaRajaeiZhang2026} proved
failure of data processing and failure of convexity of the distance
and its square.

The representation of the $z$-fidelity $F_z$, corresponding to
$\alpha=1/2$ in the $\alpha$-$z$ family
\cite{AudenaertDatta2015,NuradhaMishraLeditzkyWilde2025},
allows the polar exponent to depend on $P$ and $Q$.
At noncommuting pairs and for $z\ge1/2$, the representing exponent
on each increasing branch is unique and decreases strictly with $z$.
Below this threshold, $F_z(P,Q)$ exceeds the modulus of every
generalized fidelity of the same pair
(Theorem~\ref{thm:recover-z-from-generalized}).
Proposition~\ref{prop:qubit-explicit-base} gives an explicit
representing base in dimension two.

Although all polar exponents give the same fidelity at commuting
pairs, the representing exponent and its determinant-normalized
base extend continuously to all qubit pairs at every finite order
$z\ge1/2$. In dimensions $d\ge3$, such global extensions exist
only at the Uhlmann and Holevo orders $z=1/2$ and $z=1$
(Theorem~\ref{thm:polar-continuous-extension}).

The classification of all bases for the Holevo value reduces to
a problem involving unitary factors. Fix $P,Q$ and put
$M=P^{-1/2}Q^{1/2}$. Each base has a unitary factor $W$ satisfying
$F_R(P,Q)=\Tr(WP^{1/2}Q^{1/2})$. A unitary is attainable exactly
when $MW$ is similar to a positive definite matrix.
Positive factorization \cite{Taussky1968,JohnsonHillar2002} then
recovers the complete fiber over $W$, while a single complex
trace equation selects the Holevo factors.
As the pair varies, Ballantine's four-factor theorem
\cite{Ballantine1970,Sourour1986} determines the exact global image,
answering Afham and Ferrie's Open Problem~7
(Theorem~\ref{prop:global-central-obstruction}).

At the identity base of a noncommuting pair, the Holevo trace
equation has real rank two, giving a local manifold of bases
of dimension $d^2-2$. When $M$ has simple spectrum, the corresponding
manifold of factors at $I$ has dimension $d^2-d-2$
(Corollaries~\ref{cor:generic-holevo-base-dimension}
and~\ref{cor:generic-holevo-unitary-dimension}).
In dimension two, an explicit criterion determines whether a
second Holevo factor exists. In dimension three, choosing a left
eigenline of $MW$ reduces the simple-spectrum case to a
$2\times2$ problem. Together,
Theorems~\ref{thm:qutrit-simple} and~\ref{thm:qutrit-repeated}
give a complete classification of the factors.
The simple-spectrum locus is dense in the full factor set
of every noncommuting qutrit pair
(Corollary~\ref{cor:qutrit-exception-closure}).
\section{Preliminaries}
\label{sec:prelim}

Let $\mathbb P_d$ denote the cone of $d\times d$ positive definite complex matrices. For invertible $Y$ we write $Y=\Pol(Y)|Y|$, where $|Y|=(Y^*Y)^{1/2}$ and
\[
\Pol(Y):=Y(Y^*Y)^{-1/2}=(YY^*)^{-1/2}Y.
\]
By \cite[Eq.~(3.3)]{AfhamFerrie2024}, the generalized fidelity
in \eqref{eq:generalized-fidelity} admits the unitary-factor
representation
\begin{equation}
	F_R(P,Q)=\Tr\!\bigl[Q^{1/2}U_QU_P^*P^{1/2}\bigr],
	\label{eq:generalized-fidelity-unitary}
\end{equation}
where $U_P:=\Pol(P^{1/2}R^{1/2})$ and
$U_Q:=\Pol(Q^{1/2}R^{1/2})$.
The Uhlmann, Holevo, and Matsumoto fidelities are recovered
at particular bases:
\[
\begin{aligned}
	\FU(P,Q)&=F_P(P,Q)=F_Q(P,Q),\\
	\FH(P,Q)&=F_I(P,Q),\\
	\FM(P,Q)&=F_{P^{-1}}(P,Q)=F_{Q^{-1}}(P,Q).
\end{aligned}
\]
We refer to \cite{Uhlmann1976,Jozsa1994,NielsenChuang2010,Watrous2018}
for Uhlmann fidelity, \cite{Holevo1972,Wilde2018} for Holevo fidelity,
and \cite{Matsumoto2010,CreeSikora2020} for Matsumoto fidelity. Throughout, we use the unsquared
(root-fidelity) convention, as in Afham-Ferrie; thus
$\FU(P,Q)=\lVert P^{1/2}Q^{1/2}\rVert_1$, whereas some references call
its square the fidelity. We also use $\FM(P,Q)=\Tr(P\#Q)$, where
$P\#Q=P^{1/2}(P^{-1/2}QP^{-1/2})^{1/2}P^{1/2}$.
For $z>0$, let $F_z$ denote the half-order trace functional of
Audenaert and Datta \cite{AudenaertDatta2015},
\[
F_z(P,Q):=\Tr\!\left[\left(P^{1/(2z)}Q^{1/(2z)}\right)^z\right],
\]
where the power is the principal matrix power. For states and
$z\ge1/2$, this is the $z$-fidelity of
Nuradha et al.~\cite[Eq.~(1.3)]{NuradhaMishraLeditzkyWilde2025}.
Dinh, Le and Vuong \cite{DinhLeVuong2024} obtained variational
formulas for the two-parameter trace functional. Yan et al.\
\cite{YanSunDuTang2025} use its $1/\alpha$ power.

As $z\to\infty$, $F_z$ converges to the Log-Euclidean fidelity
\cite[Eq.~(1.8) and Appendix~A.1]{NuradhaMishraLeditzkyWilde2025},
\[
\FLE(P,Q):=\Tr\!\left[\exp\!\left(
\frac{\log P+\log Q}{2}\right)\right].
\]
The corresponding Log-Euclidean mean is discussed in
\cite{BhatiaGaubertJain2019}.
The notation $T\sim\Pd$ means that $T$ is similar to a positive
definite matrix, or equivalently, that it is diagonalizable with
strictly positive real spectrum. Set $S^{-*}:=(S^{-1})^*$ for
$S\in GL_d(\mathbb C)$ and $\Herm Z:=(Z+Z^*)/2$.

For positive definite $A,B$, the product $AB$ is similar to both
$A^{1/2}BA^{1/2}$ and $B^{1/2}AB^{1/2}$. Taking $A=P^{1/(2z)}$
and $B=Q^{1/(2z)}$ gives the two trace identities below by similarity
invariance of principal powers. Both expressions are continuous on the
positive semidefinite cone, where positive definite pairs are dense.

\begin{lemma}
	\label{lem:z-positive-representation}
	For $z>0$ and $P,Q\in\Pd$,
	\[
	F_z(P,Q)
	=\Tr\!\left[\left(P^{1/(4z)}Q^{1/(2z)}P^{1/(4z)}\right)^z\right]
	=\Tr\!\left[\left(Q^{1/(4z)}P^{1/(2z)}Q^{1/(4z)}\right)^z\right]>0.
	\]
	Either trace formula gives the unique continuous extension of $F_z$
	to positive semidefinite pairs.
\end{lemma}

The trace-norm variational formula applied to $A^{1/2}B^{1/2}$
gives the characterization of $\FU(A,B)$ below \cite{Watrous2018}.
Choosing the unitary factor from \eqref{eq:generalized-fidelity-unitary}
in this formula gives the bound for $F_R(A,B)$ \cite{AfhamFerrie2024}.

\begin{lemma}
	\label{lem:uhlmann-variational}
	For $A,B,R\in\Pd$,
	\[
	\FU(A,B)
	=\max_{U\in U(d)}
	\left|\Tr\!\left(B^{1/2}UA^{1/2}\right)\right|,
	\qquad
	|F_R(A,B)|\le\FU(A,B).
	\]
\end{lemma}

\begin{lemma}
	\label{lem:base-scaling}
	\label{lem:polar-basic}
	\label{lem:product-positive-iff-commute}
	\label{rem:reduction-to-faithful-states}
	For $a,b,c>0$ and $P,Q,R\in\Pd$,
	\[
	F_{cR}(P,Q)=F_R(P,Q),\qquad
	F_R(aP,bQ)=\sqrt{ab}\,F_R(P,Q).
	\]
	The same homogeneity in $P,Q$ holds for both one-sided polar paths,
	their average, and the fidelities $\FM,\FH,\FU,F_z,\FLE$. Thus, for
	any of these functions $\mathcal F$, setting $\rho=P/\Tr P$ and
	$\sigma=Q/\Tr Q$ gives
	\[
	\mathcal F(P,Q)
	=\sqrt{\Tr P\,\Tr Q}\,\mathcal F(\rho,\sigma).
	\]
	Consequently, pointwise inequalities between these functions extend
	from faithful states to positive definite inputs. Since quantum
	channels preserve trace, the same reduction applies to their
	data-processing inequalities whenever the outputs are positive definite.
	
	For $X\in GL_d(\mathbb C)$ and $U,V,W\in U(d)$,
	\[
	\Pol(UXV)=U\Pol(X)V,\qquad \Pol(X^*)=\Pol(X)^*,
	\]
	and
	\[
	\Pol(X)=W \iff W^*X\in\Pd \iff XW^*\in\Pd.
	\]
	For $A,B\in\Pd$, the product $AB$ is positive definite if and only
	if $A$ and $B$ commute.
\end{lemma}

\begin{proof}
 The powers $(aP)^x=a^xP^x$ and $(bQ)^x=b^xQ^x$ change the polar bases
 only by scalar multiples. All scaling identities now follow from the definitions. For generalized and named
 fidelities, see \cite[Sec.~III.A, items~10 and~13,
 Eqs.~(3.9), (3.11)-(3.13)]{AfhamFerrie2024}.
 Normalization gives the stated reduction, since channels preserve trace.
 The polar identities follow from uniqueness of polar decomposition.
 Finally, $AB$ is Hermitian exactly when $AB=BA$, and a commuting product
 is positive by simultaneous diagonalization \cite[Ch.~1]{Bhatia2007}.
\end{proof}

A quantum channel is a completely positive trace-preserving linear map
between matrix algebras. 

For $P,Q\in\Pd$ and $x\in\mathbb R$, define
\begin{align}
 \Phi_P(x)&:=F_{P^x}(P,Q)
 =\Tr\!\left[P^{\frac{1-x}{2}}
 (P^{\frac x2}QP^{\frac x2})^{1/2}\right],\\
 \Phi_Q(x)&:=F_{Q^x}(P,Q)
 =\Tr\!\left[Q^{\frac{1-x}{2}}
 (Q^{\frac x2}PQ^{\frac x2})^{1/2}\right].
\end{align}
Following Afham and Ferrie \cite[Eq.~(4.26)]{AfhamFerrie2024},
we call their arithmetic mean the $x$-Polar fidelity and write
\[
\Fpol{x}(P,Q):=\frac{\Phi_P(x)+\Phi_Q(x)}{2}.
\]
Both one-sided fidelities are positive since
$\Tr(AB)=\Tr(A^{1/2}BA^{1/2})>0$ for $A,B>0$.

The following trace inequalities follow from the log-majorization
of Bebiano, Lemos and da Provid\^encia
\cite[Theorem~2.1]{BLP2005}; see also
Liu and Cheng \cite[Eq.~(9) and Lemma~3]{LiuCheng2025}.

\begin{lemma}
	\label{lem:liu-cheng-powers}
	For $A,B\in\Pd$, $0<s<1$, and $b\in\mathbb R$, one has
	\[
	\Tr(A^{b+s}B^s)\le
	\Tr\!\left[A^b(A^{1/2}BA^{1/2})^s\right]\qquad (b\ge0).
	\]
	The inequality reverses when $b+s\le0$.
\end{lemma}

For $b>0$, apply \cite[Theorem~2.1]{BLP2005} to $A^b,B^b$ with
$p=1/b$ and $q=s/b$, and take traces. Continuity gives $b=0$.
When $b+s\le0$, apply the forward log-majorization to
$A^{-1},A^{1/2}BA^{1/2}$ with weight $-b-s$. This exchanges the
two sides and gives the reverse inequality.

\section{Polar monotonicity and data processing}
\label{sec:monotonicity}
\subsection{Variation of the base}

For $0<\alpha<1$, Afham and Ferrie
\cite[Definition~9.1]{AfhamFerrie2024} consider
\[
F_R^\alpha(P,Q)=\Tr\!\left[
(R^{1/2}PR^{1/2})^\alpha R^{-1}
(R^{1/2}QR^{1/2})^{1-\alpha}\right].
\]
At $\alpha=1/2$ this equals $F_R(P,Q)$. Along the two polar paths
the trace is positive. Afham and Ferrie identified the bases giving
the sandwiched and reverse sandwiched R\'enyi divergences
\cite[Theorem~9.2]{AfhamFerrie2024}. These bases are the extrema of
their respective paths.

\begin{theorem}[Two-range polar monotonicity]
	\label{thm:one-sided-polar-monotone}
	Let $P,Q\in\Pd$ and $0<\alpha<1$. The functions
	\[
	y\longmapsto F_{Q^y}^\alpha(P,Q),\qquad
	x\longmapsto F_{P^x}^\alpha(P,Q)
	\]
	are nondecreasing up to, and nonincreasing after, their respective
	maximizing parameters
	\[
	y_\alpha=\frac{1-\alpha}{\alpha},\qquad
	x_\alpha=\frac{\alpha}{1-\alpha}.
	\]
	If $P$ and $Q$ do not commute, both branches of each path are strictly
	monotone and each maximizer is unique. Otherwise, both paths are
	constant with value $\Tr(P^\alpha Q^{1-\alpha})$.
	The corresponding base-dependent divergences attain their minima
	at $y_\alpha$ and $x_\alpha$, respectively.
\end{theorem}

\begin{proof}
	Combining the powers of $Q$ inside the trace gives
	\[
	H(y):=F_{Q^y}^\alpha(P,Q)
	=\Tr\!\left[Q^{1-\alpha-\alpha y}
	(Q^{y/2}P Q^{y/2})^\alpha\right]>0.
	\]
	For $u<v$, put $\delta=v-u$, $A=Q^\delta$,
	$B=Q^{u/2}P Q^{u/2}$, and $b=(1-\alpha-\alpha v)/\delta$. Then
	\[
	H(u)=\Tr(A^{b+\alpha}B^\alpha),\qquad
	H(v)=\Tr\!\left[A^b(A^{1/2}BA^{1/2})^\alpha\right].
	\]
	If $v\le y_\alpha$, then $b\ge0$, and
	Lemma~\ref{lem:liu-cheng-powers} gives $H(u)\le H(v)$.
	If $y_\alpha\le u$, then
	$b+\alpha=(1-\alpha-\alpha u)/\delta\le0$,
	and the reverse inequality follows from the same lemma.
	
	For strictness, set $C=Q^{1-\alpha}$ and $D=P^\alpha$. The trace
	function
	\[
	J(p)=\Tr[(C^{p/2}D^pC^{p/2})^{1/p}]
	\]
	is strictly increasing on $(0,\infty)$ when $C,D$ do not commute
	\cite[Corollary~3.4]{FriedlandSo1994} (see also \cite{Hiai1994}).
	Since $H(0)=J(1)$ and $H(y_\alpha)=J(1/\alpha)$, we have
	$H(0)<H(y_\alpha)$ whenever $[P,Q]\ne0$.
	As $H$ is real analytic on $\mathbb R$, equality at two distinct
	points on a monotone branch would force $H$ to be constant on
	$\mathbb R$, a contradiction. For commuting inputs, the formula
	for $H$ reduces to $\Tr(P^\alpha Q^{1-\alpha})$ for every $y$.
	
	By cyclicity and the power formula along the $P$-path,
	$F_{P^x}^\alpha(P,Q)=F_{P^x}^{1-\alpha}(Q,P)$.
	Interchanging $P,Q$ and replacing $\alpha$ by $1-\alpha$ gives the
	assertion for the $P$-path.
Along either polar path, $F_R^\alpha(P,Q)$ is positive, so the
divergence reduces to $(\alpha-1)^{-1}\log F_R^\alpha(P,Q)$.
Since $\alpha<1$, it is minimized at the maximizing parameter of
the corresponding trace.
\end{proof}

Setting $\alpha=1/2$ in Theorem~\ref{thm:one-sided-polar-monotone}
gives $x_{1/2}=y_{1/2}=1$. Since the two one-sided polar paths
have the same direction of monotonicity on each side of $1$,
so does their arithmetic mean, the $x$-Polar fidelity.

\begin{corollary}
	\label{cor:x-polar-monotone}
	Let $P,Q\in\Pd$. The one-sided paths $\Phi_P,\Phi_Q$ and the
	$x$-Polar fidelity $x\mapsto\Fpol{x}(P,Q)$ are nondecreasing on
	$(-\infty,1]$ and nonincreasing on $[1,\infty)$. All three attain
	their maximum $\FU(P,Q)$ at $x=1$.
	
	If $[P,Q]\ne0$, all three functions are strictly monotone on each
	of these intervals. If $[P,Q]=0$, all three are constant.
\end{corollary}

Applying the log-majorization behind Lemma~\ref{lem:liu-cheng-powers}
with the choices of $A,B,b$ from the proof of
Theorem~\ref{thm:one-sided-polar-monotone} also compares the
eigenvalues along each branch. 
For $0<\alpha<1$ put
\[
K_\alpha(y)=Q^{(1-\alpha-\alpha y)/2}
(Q^{y/2}PQ^{y/2})^\alpha Q^{(1-\alpha-\alpha y)/2}.
\]
Writing $\prec_{\log}$ for log-majorization of the eigenvalues in
decreasing order, we obtain
\[
\begin{aligned}
	K_\alpha(u)&\prec_{\log}K_\alpha(v) &&(u<v\le y_\alpha),\\
	K_\alpha(v)&\prec_{\log}K_\alpha(u) &&(y_\alpha\le u<v).
\end{aligned}
\]

The half-order assertion answers Open Problem~4 of
\cite[Sec.~X]{AfhamFerrie2024} on the entire real line. The
$P$-based trace also identifies an overlap with the geodesic power
functionals of Capel and Costa Rico, denoted here by $Q^{\mathrm{geo}}_{\alpha,t}$.
For faithful states, inverting the positive matrix inside the power in \cite[Eq.~(84)]{CapelCostaRico2026}
gives
\begin{equation}
Q^{\mathrm{geo}}_{\alpha,t}(P\Vert Q)=
\Tr\!\left[P^{\alpha+(1-\alpha)t}
(P^{-t/2}QP^{-t/2})^{1-\alpha}\right]
=F_{P^{-t}}^\alpha(P,Q).
\label{eq:geodesic-polar-identification}
\end{equation}
Their monotonicity theorem on $0\le t\le1$ therefore applies to the
interval $-1\le x\le0$ of the $P$-based path \cite[Theorem~5.3]{CapelCostaRico2026}.

\subsection{Data processing on the real line}
Throughout this subsection, the exponent is fixed and the polar
base at the output of a quantum channel is formed from the output
pair. Write
\[
 G_x^P(P,Q)=F_{P^x}(P,Q),\qquad G_x^Q(P,Q)=F_{Q^x}(P,Q).
\]
By the one-sided trace formulas, $G_x^Q(P,Q)=G_x^P(Q,P)$.

\begin{theorem}
\label{thm:polar-dpi-classification}
Fix $x\in\mathbb R$ and choose
$G_x\in\{G_x^P,G_x^Q,\Fpol{x}\}$. The inequality
\[
 G_x(\Lambda(P),\Lambda(Q))\ge G_x(P,Q)
\]
holds for every finite-dimensional quantum channel $\Lambda$ and
every positive definite input pair with positive definite outputs
if and only if $x\in[-1,0]\cup\{1\}$.
For $x<-1$ a preparation channel on qubits gives a counterexample.
For $x>0$, $x\ne1$, a partial trace
$M_4(\mathbb C)\to M_2(\mathbb C)$ gives a counterexample.
In both cases the inputs and outputs can be chosen faithful states.
\end{theorem}

\begin{proof}
For $x=-t$ with $0\le t\le1$, the identification
\eqref{eq:geodesic-polar-identification} at $\alpha=1/2$ gives
$G_{-t}^P(P,Q)=Q^{\mathrm{geo}}_{1/2,t}(P\Vert Q)$ on faithful states.
Data processing for $-2\log Q^{\mathrm{geo}}_{1/2,t}$
\cite[Corollary~6.3]{CapelCostaRico2026} gives the desired inequality
for $G_{-t}^P$. Interchanging the inputs gives it for $G_{-t}^Q$, and
averaging gives it for $\Fpol{-t}$. Homogeneity and trace preservation
extend the inequalities to positive definite inputs. At $x=1$ all three
functions equal $\FU$, which satisfies data processing \cite{Watrous2018}.

Suppose $x<-1$, and choose any noncommuting faithful qubit states
$\rho,\sigma$. The optimal reverse test for Matsumoto fidelity
\cite{Matsumoto2010} provides a suitable preparation channel.
Diagonalize
$\sigma^{-1/2}\rho\sigma^{-1/2}=\sum_i\lambda_i|e_i\rangle\langle e_i|$
and set $v_i=\sigma^{1/2}e_i$, $q_i=\|v_i\|^2$, $p_i=\lambda_iq_i$.
Then $p,q$ are strictly positive probability vectors. The channel
\[
 \Gamma(Z)=\sum_i Z_{ii}\frac{v_i v_i^*}{q_i}
\]
sends $\operatorname{diag}p,\operatorname{diag}q$ to $\rho,\sigma$,
respectively. The common classical fidelity of the input pair is
\[
 \sum_i\sqrt{p_iq_i}
 =\Tr\!\left[\sigma^{1/2}
    (\sigma^{-1/2}\rho\sigma^{-1/2})^{1/2}\sigma^{1/2}\right]
 =\FM(\rho,\sigma).
\]
On the output pair strict polar monotonicity
(Theorem~\ref{thm:one-sided-polar-monotone}) gives
$G_x(\rho,\sigma)<G_{-1}(\rho,\sigma)=\FM(\rho,\sigma)$ for each
of the three functions. Thus the same channel violates all three inequalities.

It remains to find a violation for $x>0$, $x\ne1$. Put
$P_u=\operatorname{diag}(u,1)$ for $u>0$, let $E=vv^*$ for
$v=(1,1)^T/\sqrt2$, and put $Q_\varepsilon=E+\varepsilon I_2$.
We first compute the pointwise limit
\[
 g_x(u)=\lim_{\varepsilon\downarrow0}\Fpol{x}(P_u,Q_\varepsilon)
 =\frac{1}{2\sqrt2}
 \left(\frac{1+u^{(1+x)/2}}{\sqrt{1+u^x}}+\sqrt{1+u}\right).
\]
For the $P$-based term, substitute $Q_\varepsilon\to E$ and use
$(ww^*)^{1/2}=ww^*/\|w\|$ with $w=P_u^{x/2}v$. For the
$Q$-based term, set $H_\varepsilon=
(Q_\varepsilon^{x/2}P_uQ_\varepsilon^{x/2})^{1/2}$. Then
$H_\varepsilon\to(EP_uE)^{1/2}=\sqrt{(1+u)/2}\,E$.
The factor $Q_\varepsilon^{(1-x)/2}$ in the trace can grow on
$E^\perp$ when $x>1$. Operator monotonicity of the square root gives
$H_\varepsilon\le\|P_u\|^{1/2}Q_\varepsilon^{x/2}$, so the
contribution from $E^\perp$ is at most
$\|P_u\|^{1/2}\varepsilon^{1/2}$. On $E$ the multiplying factor
$(1+\varepsilon)^{(1-x)/2}$ tends to one. This proves the limit.

The continuous extension to $u=0$ satisfies $g_x(0)=1/\sqrt2$ and
$g_x(1)=1$. At the left endpoint the formula above gives
\[
 g_x(u)-g_x(0)=
 \begin{cases}
  -u^x/(4\sqrt2)+o(u^x),& 0<x<1,\\
  u/(4\sqrt2)+o(u),& x>1 .
 \end{cases}
\]
The right derivative at zero is $-\infty$ for $0<x<1$ and
$1/(4\sqrt2)$ for $x>1$. Both values are strictly smaller than the
secant slope $g_x(1)-g_x(0)=1-1/\sqrt2$, contradicting concavity
on $[0,1]$.
By continuity there are $u_1,u_2>0$ and $0<\lambda<1$ such that,
with $u_*=\lambda u_1+(1-\lambda)u_2$,
\[
 g_x(u_*)<\lambda g_x(u_1)+(1-\lambda)g_x(u_2).
\]
Fix these three arguments. The same strict inequality holds with
$g_x(u)$ replaced by $\Fpol{x}(P_u,Q_\varepsilon)$ for all
sufficiently small $\varepsilon>0$.

Take $A=(\lambda P_{u_1})\oplus((1-\lambda)P_{u_2})$ and
$B=(\lambda Q_\varepsilon)\oplus((1-\lambda)Q_\varepsilon)$.
The partial trace over the block index sends $A,B$ to
$P_{u_*},Q_\varepsilon$. Additivity on direct sums and homogeneity
identify the right-hand side of the strict inequality with
$\Fpol{x}(A,B)$. Hence the channel decreases $\Fpol{x}$.
It must decrease at least one of the two one-sided functions on this
pair. Interchanging $A,B$ gives a violation for the other function.
Normalizing each input by its own trace preserves all strict
inequalities, by homogeneity and trace preservation.
\end{proof}

\section{Polar bases representing other fidelities}
\label{sec:polar-consequences}

\subsection{Interval realization and pointwise representation}

By Corollary~\ref{cor:x-polar-monotone} and continuity, each of the
three paths $\Phi_P$, $\Phi_Q$ and $x\mapsto\Fpol{x}(P,Q)$ maps
$[-1,1]$ onto $[\FM(P,Q),\FU(P,Q)]$. For noncommuting pairs,
each value in this interval is attained exactly once.

For faithful states $P,Q$, this interval contains the value of every
fidelity $G$ that agrees with classical fidelity on commuting states
and satisfies data processing, by the bounds
$\FM(P,Q)\le G(P,Q)\le\FU(P,Q)$
\cite{Matsumoto2010,Matsumoto2014,Watrous2018}.

Write $F_\infty=\FLE$. Following Afham and Ferrie
\cite[Definition~5.2]{AfhamFerrie2024}, an interior fidelity at
$(P,Q)$ is a finite convex combination of generalized fidelity
values $F_R(P,Q)$.

\begin{theorem}
	\label{thm:recover-z-from-generalized}
	Let $P,Q\in\Pd$ be noncommuting. For every $z\in[1/2,\infty]$,
	there exist unique parameters
	$\theta_P(z),\theta_Q(z),\vartheta(z)\in[-1,1]$ such that
	\[
	F_z(P,Q)=\Phi_P(\theta_P(z))=\Phi_Q(\theta_Q(z))
	=\Fpol{\vartheta(z)}(P,Q).
	\]
	All three parameters are strictly decreasing in $z$, equal $1$
	at $z=1/2$ and $0$ at $z=1$, and converge as $z\to\infty$
	to their respective values for $\FLE(P,Q)$.
	
	For $0<z<1/2$, one has $F_z(P,Q)>\FU(P,Q)$.
	Consequently, $F_z(P,Q)$ cannot be realized by any generalized
	or interior fidelity, even after taking its real part.
	
	For commuting pairs, all three paths and all $F_z$,
	$z\in(0,\infty]$, are constant with the common value
	$\Tr(P^{1/2}Q^{1/2})$.
\end{theorem}
\begin{proof}
For $p>0$ put $C=P^{1/2}$, $D=Q^{1/2}$ and
\[
 J(p)=\Tr\!\left[(C^{p/2}D^pC^{p/2})^{1/p}\right].
\]
Then $F_z(P,Q)=J(1/z)$ by Lemma~\ref{lem:z-positive-representation}.
Since $C$ and $D$ do not commute, $J$ is strictly increasing
\cite[Corollary~3.4]{FriedlandSo1994}. The Lie-Trotter formula gives
$J(p)\to\FLE(P,Q)$ as $p\downarrow0$, so $F_z(P,Q)$ is strictly
decreasing on $(0,\infty]$. Its values at $z=1/2$ and $z=1$ are
$\FU(P,Q)$ and $\FH(P,Q)$.

For the lower endpoint, the weighted complementary Golden-Thompson
inequality \cite[Corollary~2.3]{AndoHiai1994} gives
\[
 \FM(P,Q)=\Tr(P\#Q)
 \le \Tr\exp\!\left(\tfrac12\log P+\tfrac12\log Q\right)
 =\FLE(P,Q).
\]
Thus for $z\in[1/2,\infty]$ the target belongs to
$[\FM(P,Q),\FU(P,Q)]$.  The inverse of each path restricted to $[-1,1]$ is continuous and
strictly increasing.  Applying
these inverses to $F_z(P,Q)$ gives the three exponents, their strict
decrease and their limits. The values at $z=1$ follow because every
path takes $\FH(P,Q)$ at $x=0$.

For $0<z<1/2$, strict decrease gives $F_z(P,Q)>F_{1/2}(P,Q)$.
Lemma~\ref{lem:uhlmann-variational} bounds every $|F_R(P,Q)|$ by
$\FU(P,Q)$. The triangle inequality gives the same bound for the
modulus of every convex combination, which proves the exclusions.
For commuting $P,Q$, simultaneous diagonalization gives
\[
 F_z(P,Q)=\Phi_P(x)=\Phi_Q(x)=\Fpol{x}(P,Q)=\Tr(P^{1/2}Q^{1/2})
\]
for every $z\in(0,\infty]$ and $x\in\mathbb R$.
\end{proof}

At $\alpha=1/2$, the same finite-order threshold $z\ge1/2$ characterizes
data processing for the $\alpha$-$z$ R\'enyi divergence
\cite[Theorem~1.2]{Zhang2020}. The representation above allows the polar
parameter to depend on the pair. For
$d\ge2$, Capel and Costa Rico \cite[Proposition~4.1]{CapelCostaRico2026}
classify the finite orders admitting a fixed parameter on a one-sided
path. Under \eqref{eq:geodesic-polar-identification} with
$\alpha=1/2$ and $t=-x$, the identity
$F_{P^x}(P,Q)=F_z(P,Q)$ holds for every pair if and only if
$(x,z)=(0,1)$ or $(1,1/2)$. Symmetry gives the same conclusion for the
$Q$-path.

\begin{corollary}
\label{cor:continuous-equivariant-polar-parameters}
On noncommuting pairs the three parameters are jointly continuous in
$(z,P,Q)$ for $1/2\le z\le\infty$, where convergence of the order
means convergence of $1/z$ with $1/\infty=0$. Each parameter is
invariant under simultaneous unitary conjugation. The representing bases
$P^{\theta_P(z)}$ and $Q^{\theta_Q(z)}$ are jointly continuous and
transform by the same conjugation as the inputs.
\end{corollary}

\begin{proof}
The paths are jointly continuous in the pair and the exponent by
continuous functional calculus. For finite $z$ the same is true of the
target $F_z(P,Q)$ by Lemma~\ref{lem:z-positive-representation}. To
see continuity at $z=\infty$ with varying inputs, fix any pair and
restrict the inputs to a compact neighborhood inside
$\Pd\times\Pd$. The continuous functions $F_n$ decrease pointwise to
the continuous $\FLE$ there, so the convergence is uniform by Dini's
theorem. Monotonicity in the real order gives the same uniform limit
for $F_z$ as $z\to\infty$.

For a convergent sequence of orders and pairs whose limiting pair is
noncommuting, choose any convergent subsequence of the corresponding
parameters on a given path. Joint continuity of the path and target
forces the limit to represent the limiting target on that path.
Uniqueness and compactness of $[-1,1]$ therefore give convergence of
the whole sequence. Unitary conjugation preserves both the path value
and the target, so uniqueness also gives the parameter invariance.
The statements about the bases now follow by functional calculus.
\end{proof}
The one-sided representation extends to every $0<\alpha<1$. For $z>0$
write
\[
Q_{\alpha,z}(P,Q)
=\Tr\!\left[\left(
P^{\alpha/(2z)}Q^{(1-\alpha)/z}P^{\alpha/(2z)}
\right)^z\right].
\]
For a fixed noncommuting pair, the $P$-polar branch
$-1\le x\le x_\alpha$ realizes $Q_{\alpha,z}(P,Q)$ exactly when
$z\ge1-\alpha$. The $Q$-polar branch $-1\le y\le y_\alpha$ realizes
it exactly when $z\ge\alpha$. On each branch the representing exponent
is unique. Thus both representations exist exactly when
$z\ge\max\{\alpha,1-\alpha\}$, which is also the data-processing
range for the $\alpha$-$z$ R\'enyi divergence at these orders
\cite[Theorem~1.2]{Zhang2020}.

At exponent $-1$ both traces equal the trace of the weighted
geometric mean
\[
P\#_tQ=P^{1/2}(P^{-1/2}QP^{-1/2})^tP^{1/2},
\qquad t=1-\alpha .
\]
The weighted complementary Golden-Thompson inequality
\cite[Corollary~2.3]{AndoHiai1994} gives
\[
\Tr(P\#_{1-\alpha}Q)
\le \Tr\exp\!\bigl(\alpha\log P+(1-\alpha)\log Q\bigr)
=\lim_{z\to\infty}Q_{\alpha,z}(P,Q).
\]
At their maxima the $P$- and $Q$-polar traces equal
$Q_{\alpha,1-\alpha}$ and $Q_{\alpha,\alpha}$, respectively.
By \cite[Corollary~3.4]{FriedlandSo1994}, applied to $P^\alpha$
and $Q^{1-\alpha}$ with power parameter $1/z$, the target is strictly
decreasing in $z$. The two thresholds and uniqueness now follow from
Theorem~\ref{thm:one-sided-polar-monotone} and continuity on the
indicated branches.
At a commuting pair $R=I$ represents all these targets.
For the half-order functional $F_z$,
Section~\ref{sec:polar-boundary} determines when the representing
exponent and its normalized base extend continuously to commuting pairs.

\subsection{An explicit base in dimension two}

For commuting pairs, the interval $[\FM(P,Q),\FU(P,Q)]$
reduces to a single point, which is realized by $R=I$.
For noncommuting pairs in dimension two, the inverse of the
increasing polar path admits the following closed-form expression.

\begin{proposition}
	\label{prop:qubit-explicit-base}
	Let $P,Q\in\mathbb P_2$ be noncommuting and let
	$F\in[\FM(P,Q),\FU(P,Q)]$. In an orthonormal eigenbasis
	of $P$, write
	\[
	P=\operatorname{diag}(r^2,t^2),\qquad
	Q=\begin{pmatrix}a&b\\ \overline b&c\end{pmatrix},
	\qquad r>t>0.
	\]
	Put $k=|b|>0$, $\delta=\sqrt{ac-k^2}$,
	$A=ra+t\delta$, $B=r\delta+tc$, and
	\[
	U=\sqrt{r^2a+t^2c+2rt\delta},\qquad
	\eta=\sqrt{U^2-F^2},\qquad
	s=\frac{Fkr-B\eta}{Fkt+A\eta}.
	\]
	Then $t/r\le s\le r/t$, and the base
	$R=\operatorname{diag}(s^2,1)$ in the chosen basis satisfies
	$F_R(P,Q)=F$. The unique representing exponent on the
	$P$-polar path in $[-1,1]$ is
	\[
	x=\frac{\log s}{\log(r/t)}.
	\]
\end{proposition}

\begin{proof}
	Write $R_u=\operatorname{diag}(u^2,1)$ for $u>0$.
	For $u=(r/t)^x$, the base $R_u$ is a positive multiple of $P^x$,
	so $F_{R_u}(P,Q)=F_{P^x}(P,Q)$. Applying the square-root formula
	\[
	H^{1/2}=\frac{H+\sqrt{\det H}\,I}
	{\sqrt{\Tr H+2\sqrt{\det H}}}
	\]
	for $H\in\mathbb P_2$ to $H=R_u^{1/2}QR_u^{1/2}$ gives
	\[
	F_{R_u}(P,Q)=\frac{Au+B}{\sqrt{D(u)}},\qquad
	D(u)=au^2+2\delta u+c.
	\]
	At $u=r/t$ this value is $U=\FU(P,Q)$. Expanding the squares gives
	\[
	D(u)\bigl(U^2-F_{R_u}(P,Q)^2\bigr)=k^2(r-tu)^2.
	\]
	By Corollary~\ref{cor:x-polar-monotone} and continuity,
	there is a unique $u\in[t/r,r/t]$ with $F_{R_u}(P,Q)=F$.
	Since $r-tu\ge0$ on this interval, the preceding identity yields
	\[
	\eta(Au+B)=Fk(r-tu).
	\]
	Solving for $u$ gives
	\[
	u=\frac{Fkr-B\eta}{Fkt+A\eta}=s,
	\]
	where $Fkt+A\eta>0$. The corresponding exponent is
	$x=\log s/\log(r/t)$.
\end{proof}

In particular, the targets $\FM(P,Q)$, $\FH(P,Q)$ and $\FU(P,Q)$
correspond to $s=t/r$, $s=1$ and $s=r/t$, respectively.
By Theorem~\ref{thm:recover-z-from-generalized}, the formula also
applies to $F_z(P,Q)$ for every $z\in[1/2,\infty]$.

\subsection{Continuous extension at commuting pairs}
\label{sec:polar-boundary}

For $1/2\le z<\infty$ and $[P,Q]\ne0$, let $\theta_z(P,Q)$ be the unique
$x\in[-1,1]$ with $\Phi_P(x)=F_z(P,Q)$. Positive multiples of a base give
the same fidelity. The representative of determinant one is
\[
 \widehat R_z(P,Q)=\widehat P^{\,\theta_z(P,Q)},\qquad
 \widehat P=(\det P)^{-1/d}P.
\]

\begin{theorem}
\label{thm:polar-continuous-extension}
For $d=2$, both $\theta_z$ and $\widehat R_z$ extend uniquely and
continuously to all positive definite pairs, jointly in $(z,P,Q)$ for
$1/2\le z<\infty$. For each $d\ge3$ and each such fixed $z$, each map
admits a continuous extension to all positive definite pairs if and only
if $z=1/2$ or $z=1$. At these orders the extended exponent and base are
$(1,\widehat P)$ and $(0,I)$, respectively. The same assertions hold
on the $Q$-polar path after interchanging $P$ and $Q$.
\end{theorem}

\begin{proof}
	The qubit case relies on an equation for the representing exponent
	that remains valid at commuting pairs. By homogeneity, assume
	$\det P=\det Q=1$.  Write the eigenvalues of $P$ and $Q$ as
	$e^{\pm2u}$ and $e^{\pm2v}$, respectively, with $u,v\ge0$.
	When $uv>0$, let $E_P,E_Q$ be the spectral projections for the
	larger eigenvalues and put $c=2\Tr(E_PE_Q)-1$. Then $|c|\le1$,
	with equality precisely when $P$ and $Q$ commute. When $uv=0$,
	we may assign $c$ any value in $[-1,1]$.
	
Put $s=1/z$. Define $S_k(w)=\sinh(kw)/\sinh(2w)$ for $w\ne0$,
with $S_k(0)=k/2$.
	For $T(t)=\cosh(2z\operatorname{arcosh}t)$, set
	\[
	a=\cosh(s(u-v)),\qquad b=\cosh(s(u+v)),\qquad
	y=\tfrac{1-c}{2}a+\tfrac{1+c}{2}b.
	\]
	Write $T[a,y,b]$ for the second divided difference, with repeated
	nodes interpreted by continuity. This divided difference is jointly
	continuous in the nodes and in $z$, including when the nodes approach
	$1$. Indeed, $w\mapsto\cosh\sqrt w$ has derivative $1/2$ at zero
	and hence an analytic local inverse near $1$.
	
	Suppose that $P,Q$ do not commute and let $U=\FU(P,Q)$.
The spectral decompositions of $P$ and $Q$ give
\[
\frac12\Tr\!\left(P^{1/(2z)}Q^{1/(2z)}\right)
=\cosh(su)\cosh(sv)+c\sinh(su)\sinh(sv)=y.
\]
The product has determinant one and positive eigenvalues, so its
eigenvalues are $e^{\pm\xi}$ with $\xi=\operatorname{arcosh}y$.
Thus $F_z=2\cosh(z\xi)$ and $F_z^2=2+2T(y)$. By
	Proposition~\ref{prop:qubit-explicit-base}, $U^2$ is affine in $c$.
	At $c=\pm1$ it agrees with $F_z^2$, since the inputs then commute.
	Linear interpolation at $a,b$ gives
	\[
	U^2-F_z^2
	=2(y-a)(b-y)T[a,y,b]
	=2(1-c^2)\sinh^2(su)\sinh^2(sv)T[a,y,b].
	\]
	In a $P$-eigenbasis ordered by decreasing eigenvalue, the diagonal
	entries of $Q$ are $\cosh(2v)\pm c\sinh(2v)$, and
	$|Q_{12}|^2=(1-c^2)\sinh^2(2v)$. Substituting $r=e^u$, $t=e^{-u}$
	and $(r/t)^x=e^{2xu}$ in the squared-gap identity from the proof of
	Proposition~\ref{prop:qubit-explicit-base} gives
	\[
	U^2-\Phi_P(x)^2
	=\frac{2(1-c^2)\sinh^2(2v)\sinh^2((1-x)u)}{D_x},
	\]
	where
	\[
	D_x=1+\cosh(2xu)\cosh(2v)+c\sinh(2xu)\sinh(2v)\ge2.
	\]
	At $x=\theta_z(P,Q)$ the two gaps agree. Cancelling the common
	factors yields
	\[
	\frac{S_{1-x}(u)^2}{D_x}
	=S_s(u)^2S_s(v)^2T[a,y,b].
	\]
	Both sides are continuous in $u,v,c,x,z$, and when $uv=0$ they
	are independent of $c$.
	
	Uniqueness of the limiting exponent follows from strict monotonicity
	of the left-hand side at commuting pairs. When $u>0$, take  $c=\sigma\in\{-1,1\}$,
	with either sign allowed if $v=0$. Then $D_x=2\cosh^2(xu+\sigma v)$.
	The nonnegative square root of the left-hand side is strictly
	decreasing on $[-1,1]$, since
	\[
	\frac{d}{dx}\frac{\sinh((1-x)u)}{\cosh(xu+\sigma v)}
	=-\frac{u\cosh(u+\sigma v)}{\cosh^2(xu+\sigma v)}<0.
	\]
	If $u=0$, the left-hand side is a positive multiple of $(1-x)^2$.
	Thus the equation has at most one root in $[-1,1]$ at every
	commuting pair.
	
	Noncommuting qubit pairs are dense. Along any sequence of such pairs
	approaching a commuting pair, with orders tending to a finite
	$z\ge1/2$, compactness gives subsequential limits of both the exponent
	and $c$. If $uv>0$ at the limit, the spectral projections determine
	the limiting $c$. If $uv=0$, the equation is independent of $c$.
	Every subsequential limit of the exponent therefore solves the same
	equation. This proves existence of a root and convergence to its
	unique value. Assigning that root at each commuting pair gives the
	unique extension. The same argument applies to sequences that include
	commuting pairs, since their assigned values also satisfy the equation.
	Together with Corollary~\ref{cor:continuous-equivariant-polar-parameters},
	this proves joint continuity in $(z,P,Q)$. Joint continuity of
	$\widehat R_z=\widehat P^{\,\theta_z}$ follows by functional calculus.
	
	In dimensions $d\ge3$, two different qubit perturbations of the same
	commuting pair suffice to detect the failure of continuous extension.  At $z=1/2$ and $z=1$, Theorem~\ref{thm:recover-z-from-generalized}
	gives the constant exponents $1$ and $0$, with normalized bases
	$\widehat P$ and $I$, respectively. Fix any other finite $z>1/2$,
	so that $0<s<2$ and $s\ne1$. Let $x_*$ be the extended qubit
	exponent at $(I_2,I_2)$. For $u>0$, take
	\[
	P_u=\operatorname{diag}(1,1,e^{4u}),\qquad Q_u=P_u^{-1}
	\]
	in dimension three, adjoining identity blocks to both matrices
	when $d>3$.
	
	First perturb only coordinates $1,2$ through noncommuting qubit
	pairs, and then consider perturbations supported instead on
	coordinates $1,3$. Both $F_z$ and $\Phi_P(x)$ are additive under
	direct sums, and their contributions from the fixed diagonal
	complement agree for every $x$. Thus the full representing exponent
	equals that of the perturbed block. For coordinates $1,2$ its limit
	is $x_*$. Denote the limit for coordinates $1,3$ by $x(u)$.
	After determinant normalization, the latter block has equal gaps
	$u=v$ and opposite eigenvalue order, so $c=-1$.
	
	For this block $a=y=1$, $b=\cosh(2su)$ and $T'(1)=4/s^2$. Hence
	\[
	T[1,1,b]=\frac{\cosh(4u)-1-(4/s^2)(b-1)}{(b-1)^2}.
	\]
Substituting this expression, $D_x=2\cosh^2((1-x)u)$ and
$b-1=2\sinh^2(su)$ into the qubit equation gives
\[
\tanh^2((1-x(u))u)
=1-\frac{4\sinh^2(su)}{s^2\sinh^2(2u)}.
\]
Using $1-\tanh^2 w=1/\cosh^2 w$ and taking the positive square
root, we obtain
\[
\cosh((1-x(u))u)=\frac{s\sinh(2u)}{2\sinh(su)}.
\]
	If $x(u)=x_*$ for every $u>0$, the right-hand side is asymptotic
	to $(s/2)e^{(2-s)u}$ as $u\to\infty$. Since $2-s>0$ and
	$x_*\in[-1,1]$, this forces $1-x_*>0$. Comparing the exponential
	rates gives $1-x_*=2-s$, and comparing the leading coefficients
	then gives $1/2=s/2$, contrary to $s\ne1$.
	
	Choose $u>0$ for which $x(u)\ne x_*$. The two perturbations yield
	distinct limiting exponents at $(P_u,Q_u)$. Since $P_u$ is nonscalar,
	$\widehat P_u$ has an eigenvalue different from $1$, so the map
	$x\mapsto\widehat P_u^{\,x}$ is injective. The two limiting
	normalized bases are therefore distinct as well. This proves the
	necessity of the two stated orders. Interchanging $P$ and $Q$ gives
	the assertions for the $Q$-polar path.
\end{proof}

\section{Holevo bases and unitary factors}
\label{sec:holevo-unitaries}

At $z=1$, the polar paths give the Holevo value at $R=I$.
The problem now is to determine all bases giving this value
for a fixed pair $P,Q\in\Pd$. Write
\[
 M=P^{-1/2}Q^{1/2},\qquad X=P^{1/2}Q^{1/2}=PM.
\]
The unitary factor of a base $R$ is
$W=U_QU_P^*$, with $U_P=\Pol(P^{1/2}R^{1/2})$ and
$U_Q=\Pol(Q^{1/2}R^{1/2})$. By \eqref{eq:generalized-fidelity-unitary},
$F_R(P,Q)=\Tr(WX)$. Among the factors $W$ that occur, the Holevo
equation selects those satisfying $\Tr(WX)=\Tr X$.

\subsection{Attainable factors and their fibers}
\label{sec:holevo-fixed-pair}
\label{sec:unitary-factors}

Let $\mathcal U_{\mathrm{att}}(P,Q)$ denote the set of unitary factors
of the bases $R\in\Pd$, and let $\mathcal R_W(P,Q)$ be the set of bases
whose factor is $W$. We use the positive coordinate
\[
 A=(P^{1/2}RP^{1/2})^{1/2},\qquad
 R_A=P^{-1/2}A^2P^{-1/2}.
\]
These formulas are mutually inverse real-analytic maps of $\Pd$.

\begin{theorem}
\label{thm:complete-classification-holevo-bases}
\label{thm:unitary-factor-classification-fixed-pair}
\label{thm:complete-structural-classification-holevo-bases}
The factor of $R_A$ is $W_A=\Pol(AM)^*$, and
$F_{R_A}(P,Q)=\Tr(W_AX)$. For $W\in U(d)$,
\[
 W\in\mathcal U_{\mathrm{att}}(P,Q)
 \quad\Longleftrightarrow\quad
 AMW>0\ \text{for some } A\in\Pd
 \quad\Longleftrightarrow\quad MW\sim\Pd .
\]
In particular, attainability is equivalent to diagonalizability of $MW$
with strictly positive real spectrum.

If $MW=SDS^{-1}$ with $D>0$, then the positive coordinates and the bases
over $W$ are exactly
\[
 \begin{aligned}
 \mathcal A_W(P,Q)&=\{S^{-*}CS^{-1}: C>0,\ [C,D]=0\},\\
 \mathcal R_W(P,Q)&=\{R_A:A\in\mathcal A_W(P,Q)\}.
 \end{aligned}
\]
Both fibers have real dimension $\sum_jm_j^2$, where $m_j$ are the
multiplicities of the distinct eigenvalues of $D$. Imposing $\det A=1$
reduces this dimension by one.

Consequently, the set $\mathcal H(P,Q)$ of bases with
$F_R(P,Q)=\FH(P,Q)$ is the disjoint union
\[
 \mathcal H(P,Q)=\bigsqcup_{W\in\mathcal U_H(P,Q)}\mathcal R_W(P,Q),
 \qquad
 \mathcal U_H(P,Q)=\{W\in U(d):MW\sim\Pd,\ \Tr(WX)=\Tr X\}.
\]
\end{theorem}

\begin{proof}
Since $|AP^{-1/2}|=R_A^{1/2}$, its polar decomposition has the form
$AP^{-1/2}=V^*R_A^{1/2}$. Hence
\[
 R_A^{1/2}P^{1/2}=VA,\qquad R_A^{1/2}Q^{1/2}=VAM,
\]
so $U_P=V^*$ and $U_Q=\Pol(AM)^*V^*$. This proves the factor and
trace formulas. Moreover,
\[
 W_A=W\quad\Longleftrightarrow\quad WAM>0
 \quad\Longleftrightarrow\quad AMW>0,
\]
where the last two matrices are unitarily conjugate.

Set $T=MW$. The relation $AT>0$ gives
$A^{1/2}TA^{-1/2}=A^{-1/2}(AT)A^{-1/2}>0$. Conversely, given
$T=SDS^{-1}$ with $D>0$, put $C=S^*AS$. Then
$AT=S^{-*}CDS^{-1}$, so $AT>0$ exactly when $CD>0$. For $C,D>0$,
the last condition is equivalent to $CD=DC$. Taking $C=I$ proves
existence, and allowing all positive $C$ in the commutant proves the
fiber formulas. This is the positive-factorization criterion of Taussky
and the corresponding factor parametrization of Johnson and Hillar
\cite{Taussky1968,JohnsonHillar2002}. The positive commutant is a product
of positive definite cones of sizes $m_j$ and thus has dimension
$\sum_jm_j^2$. Each positive ray meets $\det A=1$ transversally once.
Finally, $\FH(P,Q)=\Tr X$, and each base has a unique factor.
\end{proof}

Afham and Ferrie's Open Problem~7 asks whether every element of $SU(d)$
occurs for a single triple $(P,Q,R)$. Lawson and Lim showed that the polar
factors of products of two positive definite matrices generate $SU(d)$
\cite[Remark~5.2]{LawsonLim2008}.

\begin{theorem}[Global image of the unitary factor]
\label{prop:global-central-obstruction}
\label{cor:attainable-unitaries-special}
For $d\ge2$, a unitary matrix occurs for some positive definite triple
if and only if it is a product of four positive definite matrices. Hence
\[
 \bigcup_{P,Q\in\Pd}\mathcal U_{\mathrm{att}}(P,Q)
 =SU(d)\setminus\{\omega I:\omega^d=1,\ \omega\ne1\}.
\]
\end{theorem}

\begin{proof}
Let $\mathcal C=\{HK:H,K\in\Pd\}$. By positive factorization,
$\mathcal C=\{T:T\sim\Pd\}$. As $P,Q$ vary, $M=P^{-1/2}Q^{1/2}$
ranges over $\mathcal C$, and $\mathcal C^{-1}=\mathcal C$.
Theorem~\ref{thm:unitary-factor-classification-fixed-pair} therefore
gives attainability for some pair precisely when $M,MW\in\mathcal C$
for some $M$, or equivalently $W\in\mathcal C^{-1}\mathcal C
=\mathcal C\mathcal C$.

The reverse implication also provides a realizing base explicitly.
Given $W=H_1H_2H_3H_4$ with all $H_j>0$, take
\[
 P=H_2^2,\qquad Q=H_1^{-2},\qquad A=H_3^{-1}.
\]
Then $MW=H_3H_4$ and $AMW=H_4>0$, so the base
$R_A=P^{-1/2}H_3^{-2}P^{-1/2}$ realizes $W$.

By Ballantine's four-factor theorem \cite{Ballantine1970}, also proved
in \cite[Theorem~2(a)]{Sourour1986}, the matrices with such a
factorization are precisely those with positive real determinant,
apart from scalar matrices $\lambda I$ with $\lambda\notin(0,\infty)$.
A unitary matrix with positive real determinant has determinant one,
as in \cite[Proposition~5.1]{AfhamFerrie2024}. The only positive
scalar unitary is $I$.
\end{proof}

\subsection{The identity fiber and common-polar bases}
\label{sec:polar-slice}

Afham and Ferrie proposed equality of the two polar factors as a
sufficient condition for the Holevo value and pointed to the midpoint
equations of Lawson and Lim \cite[Sec.~X, item~5]{AfhamFerrie2024}.
The next formula recovers every base satisfying this condition.

\begin{corollary}
  \label{thm:holevo-bases-classification}
  \label{cor:polar-slice-holevo}
  \label{cor:explicit-holevo-bases}
  Set $H:=P^{-1/4}Q^{1/2}P^{-1/4}$.  For $A\in\Pd$, the following
  are equivalent:
  \begin{enumerate}
    \item[(i)]
    \[
    \Pol(R_A^{1/2}P^{1/2})
    =
    \Pol(R_A^{1/2}Q^{1/2});
    \]
    \item[(ii)] $W_A=I$;
    \item[(iii)] $AM>0$;
    \item[(iv)] there exists $B\in\Pd$ such that
    \[
    A=P^{1/4}BP^{1/4},
    \qquad [B,H]=0.
    \]
  \end{enumerate}
  Consequently, the common-polar bases are exactly
  \begin{equation}
  R=P^{-1/4}BP^{1/2}BP^{-1/4},
  \qquad B\in\Pd,\quad [B,H]=0,
  \label{eq:common-polar-base-fiber}
  \end{equation}
  and every such base satisfies $F_R(P,Q)=\FH(P,Q)$.
\end{corollary}

\begin{proof}
  The polar factors in (i) are $U_P^*$ and $U_Q^*$. Their equality
  is equivalent to $W_A=I$, hence to $AM>0$ by
  Theorem~\ref{thm:unitary-factor-classification-fixed-pair}. Write
  $B=P^{-1/4}AP^{-1/4}$.  Since
  \[
  AM=P^{1/4}BHP^{1/4},
  \]
  one has $AM>0$ if and only if $BH>0$, equivalently $[B,H]=0$.
  Substitution into $R_A$ gives \eqref{eq:common-polar-base-fiber}; when
  $W_A=I$, the unitary-factor formula gives
  $F_{R_A}(P,Q)=\Tr(P^{1/2}Q^{1/2})$.
\end{proof}

In particular, every positive function $f:\sigma(H)\to(0,\infty)$ gives
$B=f(H)$ in \eqref{eq:common-polar-base-fiber}, hence a Holevo base
\[
R_f=P^{-1/4}f(H)P^{1/2}f(H)P^{-1/4}.
\]
For $f(x)=x^t$, write $R_t=P^{-1/4}H^tP^{1/2}H^tP^{-1/4}$,
$t\in\mathbb R$.  Substitution gives
\[
R_0=I,\qquad R_1=P^{-1/2}QP^{-1/2},\qquad
R_{-1}=Q^{-1/2}PQ^{-1/2}.
\]
The full commutant family contains infinitely many distinct bases for
$d\ge2$: the map $B\mapsto R_B$ is injective because
$(P^{1/2}R_BP^{1/2})^{1/2}=P^{1/4}BP^{1/4}$, and one may take
$B=e^{tH}$ when $H$ is nonscalar.  When $H$ is scalar, every $B>0$
commutes with it; the power subfamily gives only scalar multiples of $I$,
whereas the full family also contains nonscalar bases.

The identity fiber need not exhaust the Holevo-base locus, as the following
explicit qubit example shows.

\begin{example}
  \label{ex:nonpolar-holevo-base}
  Let
  \[
  P=\begin{pmatrix}4&0\\0&1\end{pmatrix},
  \qquad
  Q=\begin{pmatrix}2&1\\1&2\end{pmatrix},
  \qquad
  R_r=\begin{pmatrix}r&0\\0&1\end{pmatrix},
  \qquad r>0.
  \]
  The positive $2\times2$ square-root formula gives
  \[
  F_{R_r}(P,Q)
  =
  \frac{(4+\sqrt3)\sqrt r+2(1+\sqrt3)}
  {\sqrt{2r+2+2\sqrt{3r}}},
  \qquad
  \FH(P,Q)=\frac32(1+\sqrt3).
  \]
  Writing $u=\sqrt r$, equality holds exactly when
  \[
  (u-1)\bigl((1-\sqrt3)u+(2+\sqrt3)\bigr)=0.
  \]
  Thus
  \[
  r=13+\frac{15\sqrt3}{2}
  \]
  gives a nonidentity Holevo base.  Here $R_r^{1/2}P^{1/2}>0$,
  whereas $R_r^{1/2}Q^{1/2}$ is not Hermitian: its off-diagonal entries
  are $\sqrt r(\sqrt3-1)/2$ and $(\sqrt3-1)/2$.
  Hence its polar factor is not $I$, so this base lies outside the
common-polar fiber. Since $R_r=P^{\log r/\log 4}$ and $r>4$, the
same $P$-polar path attains $\FH(P,Q)$ both at $x=0$ and at an
exponent $x>1$. This illustrates the restriction to $[-1,1]$ in
the uniqueness assertion of Theorem~\ref{thm:recover-z-from-generalized}.
\end{example}

We will use the following observation twice. If $Y>0$ and $W$ is unitary, then
\begin{equation}\label{eq:positive-unitary-trace}
 \Tr(WY)=\Tr Y\quad\Longrightarrow\quad W=I.
\end{equation}
Indeed, in an eigenbasis of $Y$ the real part of $\Tr Y-\Tr(WY)$ is
$\sum_j\lambda_j(1-\operatorname{Re}W_{jj})$. Every summand is nonnegative
and every $\lambda_j$ is positive. Equality forces $W_{jj}=1$ for all $j$,
so unitarity gives $W=I$.

\begin{proposition}
  \label{prop:universal-holevo-base}
  A fixed base $R\in\Pd$ satisfies
  \[
  F_R(P,Q)=\FH(P,Q)\qquad(P,Q\in\Pd)
  \]
  if and only if $R=cI$ for some $c>0$.
\end{proposition}

\begin{proof}
  If $R=cI$, the conclusion follows from base scaling.  Conversely, assume
  the displayed identity and set $P=I$.
  The unitary-factor formula and
  \eqref{eq:positive-unitary-trace} force
  $\Pol(Q^{1/2}R^{1/2})=I$ for every $Q\in\Pd$.  Hence $Q^{1/2}$ commutes
  with $R^{1/2}$ for every $Q\in\Pd$, so $R$ is scalar.
\end{proof}

\begin{corollary}
  \label{cor:commuting-pairs-only-I}
  If $[P,Q]=0$, then
  \[
  \mathcal U_H(P,Q)=\{I\},
  \qquad
  \mathcal H(P,Q)=\mathcal R_I(P,Q).
  \]
  Equivalently, for commuting pairs the polar slice exhausts all Holevo bases.
\end{corollary}

\begin{proof}
  Since $X=P^{1/2}Q^{1/2}>0$, the Holevo equation $\Tr(WX)=\Tr X$
  forces $W=I$ by \eqref{eq:positive-unitary-trace}.
\end{proof}

\subsection{Positive coordinates and the differential of the factor map}

The variable $T_A=MW_A$ is convenient for the local dimension
calculations in Section~\ref{sec:low-dimensional-and-local-geometry}.
It has a fixed left Gram matrix $G=MM^*$ and satisfies the
symmetrizer relation $AT_A=T_A^*A$
\cite{Mostafazadeh2008}.

\begin{proposition}
\label{thm:global-positive-cone-resolution}
\label{prop:positive-cone-rank-and-connectedness}
For $A\in\Pd$,
\begin{equation}
 T_A=A^{-1}(AGA)^{1/2},\qquad T_AT_A^*=G,\qquad AT_A>0 .
 \label{eq:global-positive-cone-resolution}
\end{equation}
Consequently,
\begin{equation}
 \mathcal U_{\mathrm{att}}(P,Q)=\{M^{-1}T_A:A\in\Pd\},
 \label{eq:positive-cone-surjection}
\end{equation}
and
\begin{equation}
 \mathcal H(P,Q)=\{R_A:A\in\Pd,\ \Tr(P(T_A-M))=0\}.
 \label{eq:global-positive-cone-holevo-bases}
\end{equation}
If the distinct eigenvalues of $T_A$ have multiplicities $m_j$, then
\[
 \ker\mathrm D_AW_A=\{K=K^*:KT_A=T_A^*K\}
\]
and
\begin{equation}
 \operatorname{rank}_{\mathbb R}\mathrm D_AW_A=d^2-\sum_jm_j^2 .
 \label{eq:positive-cone-resolution-rank}
\end{equation}
The rank is the same on the section $\det A=1$, and
$\mathcal U_{\mathrm{att}}(P,Q)$ is path connected.
\end{proposition}

\begin{proof}
Put $T=MW_A$. Unitarity gives $TT^*=G$, and
Theorem~\ref{thm:unitary-factor-classification-fixed-pair} gives $AT>0$.
It follows that $(AT)^2=AGA$, proving the first formula.
Since $X=PM$, cyclicity gives $\Tr(W_AX)=\Tr(PT)$, as claimed.

The map $A\mapsto T_A$ is real analytic. Differentiating $AT=T^*A$
shows that $\mathrm D_AT_A[K]=0$ implies $KT=T^*K$. Conversely, if
$KT=T^*K$, set $A_t=A+tK$. For small real $t$, both $A_t$ and
$A_tT=AT+tKT$ are positive definite. Thus $(A_tT)^2=A_tGA_t$ and
$T_{A_t}=T$. This proves the kernel formula because $W_A=M^{-1}T_A$.
In coordinates $T=SDS^{-1}$, the kernel condition reads
$[S^*KS,D]=0$. Its real dimension is $\sum_jm_j^2$, giving the stated rank.
At $\det A=1$, the radial direction $A$ complements the tangent
hyperplane $\Tr(A^{-1}K)=0$ and lies in the kernel, so restriction
preserves rank. Finally, the attainable set is the continuous image
of the path-connected cone $\Pd$ under $A\mapsto W_A$.
\end{proof}

The pairs $(W,A)$ with $W^*W=I$, $A=A^*>0$ and $AMW=(AMW)^*>0$
form a semialgebraic set in real and imaginary coordinates. Projecting
onto $W$ gives $\mathcal U_{\mathrm{att}}(P,Q)$, and imposing the trace
equation gives $\mathcal U_H(P,Q)$. Both sets are therefore semialgebraic.
The simple-spectrum part of $\mathcal U_H(P,Q)$ is also semialgebraic,
since it is specified by the nonvanishing of the discriminant of $MW$.
Each of these sets has finitely many connected components
\cite{BochnakCosteRoy1998}.

\section{Qubit and qutrit factors and local geometry}
\label{sec:low-dimensional-and-local-geometry}

The classification of low-dimensional Holevo factors is conveniently
expressed in terms of $T=MW$. Since $G=MM^*$, the inverse
$W=M^{-1}T$ is unitary exactly when $TT^*=G$.
With $h=\Tr(PM)=\FH(P,Q)$,
Theorem~\ref{thm:unitary-factor-classification-fixed-pair}
and the identity $\Tr(WX)=\Tr(PT)$ identify the Holevo factors
bijectively with the solutions of
\begin{equation}\label{eq:fixed-gram-target}
	TT^*=G,\qquad T\sim\Pd,\qquad \Tr(PT)=h.
\end{equation}

\subsection{Explicit qubit classification}
\label{sec:low-dimensional-holevo-unitaries}

For a commuting pair, Corollary~\ref{cor:commuting-pairs-only-I} gives the
unique factor $I$. For a noncommuting pair, the trace equation and reality
of the spectrum reduce the problem to a line meeting a circle. For formulas evaluating
generalized fidelity on qubit triples, see also \cite{MaRajaeiZhang2026}.

\begin{theorem}
\label{thm:dimension-two-holevo-unitaries}
  Let $P,Q\in\mathbb P_2$ and suppose that $[P,Q]\ne0$.  There is a unitary
  $V$ such that
  \[
  V^*PV=\begin{pmatrix}r^2&0\\0&t^2\end{pmatrix},
  \qquad
  V^*Q^{1/2}V=\begin{pmatrix}a&c\\c&b\end{pmatrix},
  \]
  where $r,t,a,b,c>0$, $r\ne t$, and $ab>c^2$.  Set
  \[
  h:=ra+tb,
  \qquad
  k:=c(t-r),
  \qquad
  \ell:=\frac ar+\frac bt,
  \qquad
  n:=c\left(\frac1t-\frac1r\right),
  \qquad
  \delta:=\frac{ab-c^2}{rt},
  \]
  and define
  \[
  u_*:=\frac{h^2-k^2}{h^2+k^2},
  \qquad
  v_*:=\frac{2hk}{h^2+k^2},
  \qquad
  W_*:=\begin{pmatrix}u_*&v_*\\-v_*&u_*\end{pmatrix}.
  \]
  Finally, put
  \[
  \tau_*:=\ell u_*+nv_*
  =\frac{\ell(h^2-k^2)+2nhk}{h^2+k^2}.
  \]
  Then
  \[
  \mathcal U_H(P,Q)
  =
  \begin{cases}
  \{I\},&\tau_*\le2\sqrt\delta,\\[2mm]
  \{I,VW_*V^*\},&\tau_*>2\sqrt\delta.
  \end{cases}
  \]
 
\end{theorem}

\begin{proof}
  Since $P$ and $Q$ do not commute, $P$ has two distinct eigenvalues.  We may
  therefore diagonalize $P$ and then use a diagonal unitary to make the
  off-diagonal entry of $Q^{1/2}$ strictly positive.  This gives the stated
  canonical form.  Simultaneous unitary conjugation sends
  $\mathcal U_H(P,Q)$ onto
  $\mathcal U_H(V^*PV,V^*QV)$, so it suffices to work in this canonical form.

  Write
  \[
  M=P^{-1/2}Q^{1/2}
  =\begin{pmatrix}a/r&c/r\\c/t&b/t\end{pmatrix},
  \qquad
  X=P^{1/2}Q^{1/2}
  =\begin{pmatrix}ra&rc\\tc&tb\end{pmatrix}.
  \]
  If $W\in\mathcal U_H(P,Q)$, then
  Theorem~\ref{prop:global-central-obstruction} gives $W\in SU(2)$.
  We can consequently write
  \[
  W=\begin{pmatrix}\alpha&\beta\\-\overline\beta&\overline\alpha\end{pmatrix},
  \qquad
  \alpha=u+ix,
  \qquad
  \beta=v+iy,
  \qquad
  u^2+v^2+x^2+y^2=1.
  \]
  The imaginary part of the Holevo equation $\Tr(WX)=\Tr X$ is
  \begin{equation}
  (ra-tb)x+c(r+t)y=0.
  \label{eq:d2-holevo-imaginary}
  \end{equation}
  Moreover, $MW$ has positive real eigenvalues, so $\Tr(MW)$ is real.  Its
  imaginary part gives
  \begin{equation}
  \left(\frac ar-\frac bt\right)x
  +c\left(\frac1r+\frac1t\right)y=0.
  \label{eq:d2-spectrum-imaginary}
  \end{equation}
The determinant of the coefficient matrix in
\eqref{eq:d2-holevo-imaginary} and \eqref{eq:d2-spectrum-imaginary} is
  \[
  \frac{c(a+b)(r+t)(r-t)}{rt},
  \]
  which is nonzero.  Thus $x=y=0$, and every element of
  $\mathcal U_H(P,Q)$ is a real rotation
  \[
  W=\begin{pmatrix}u&v\\-v&u\end{pmatrix},
  \qquad u^2+v^2=1.
  \]
  The real part of the Holevo equation is
  \[
  h(u-1)+kv=0.
  \]
  The intersection of this line with the unit circle consists of
  $(u,v)=(1,0)$ and $(u,v)=(u_*,v_*)$.  Hence the only possible Holevo
  unitaries are $I$ and $W_*$.

  It remains to decide when $W_*$ is attainable.  We have
  \[
  \det(MW_*)=\det M=\delta,
  \qquad
  \Tr(MW_*)=\ell u_*+nv_*=\tau_*.
  \]
  Thus the characteristic polynomial of $MW_*$ is
  \[
  \lambda^2-\tau_*\lambda+\delta.
  \]
  If $\tau_*>2\sqrt\delta$, this polynomial has two distinct positive
  roots, so $MW_*$ is diagonalizable with positive spectrum and hence is
  similar to a positive definite matrix.  Conversely, if $MW_*$ is similar
  to a positive definite matrix, then $\tau_*\ge2\sqrt\delta$.  Equality
  would force both eigenvalues to be $\sqrt\delta$; diagonalizability would
  then give $MW_*=\sqrt\delta I$.  It would follow that
  $MM^*=\delta I$, or equivalently $Q=\delta P$, contrary to $[P,Q]\ne0$.
  Therefore attainability is equivalent to $\tau_*>2\sqrt\delta$.
  The identity always belongs to $\mathcal U_H(P,Q)$, and conjugating back by
  $V$ completes the proof.
\end{proof}

For a noncommuting pair, every Holevo factor in dimension two
gives a matrix $MW$ with simple spectrum. Hence each corresponding
fiber has real dimension two by
Theorem~\ref{thm:unitary-factor-classification-fixed-pair}.
Both possible cardinalities occur.
Example~\ref{ex:nonpolar-holevo-base} gives a second factor,
whereas for $r=1$, $t=4$, $a=c=1$ and $b=2$,
one has $\tau_*=3/4<2\sqrt\delta=1$, so the identity
is the only factor.

\subsection{The local dimension at a simple spectrum}
\label{sec:holevo-local-geometry}

The simple-spectrum condition permits a local description of both
the attainable and Holevo factor sets.

For $P,Q\in\Pd$, set
\[
M:=P^{-1/2}Q^{1/2},\qquad X:=P^{1/2}Q^{1/2}.
\]
Let $W_0\in\mathcal U_{\mathrm{att}}(P,Q)$ be such that
$T_0:=MW_0$ has distinct positive eigenvalues
$\mu_1,\ldots,\mu_d$, with spectral idempotents $E_1,\ldots,E_d$.
Write
\[
\mathfrak{su}(d)
:=\{K\in M_d(\mathbb C):K^*=-K,\ \Tr K=0\}
\]
and define
\[
\mathfrak t_{T_0}
:=
\left\{K\in\mathfrak{su}(d):
\operatorname{Im}\Tr(E_jK)=0,\quad 1\le j\le d
\right\}.
\]
Define the real-linear map
\begin{equation}
	L_{X,W_0}:\mathfrak t_{T_0}\longrightarrow\mathbb C,
	\qquad
	L_{X,W_0}(K):=\Tr(KXW_0).
	\label{eq:global-simple-holevo-transversality}
\end{equation}
At $W_0=I$, provided that $M$ has simple spectrum, we write
$\mathfrak t_M$ and $L_X(K):=\Tr(KX)$.

\begin{theorem}
	\label{thm:local-holevo-unitary-geometry}
	Let $W_0\in\mathcal U_{\mathrm{att}}(P,Q)$ be such that $MW_0$
	has distinct positive eigenvalues. Then there is an open neighborhood
	$\mathcal N$ of $W_0$ in $SU(d)$ such that
	$\mathcal U_{\mathrm{att}}(P,Q)\cap\mathcal N$ is a real-analytic
	embedded submanifold with
	\[
	T_{W_0}\mathcal U_{\mathrm{att}}(P,Q)=W_0\mathfrak t_{T_0},
	\qquad
	\dim_{\mathbb R}
	\bigl(\mathcal U_{\mathrm{att}}(P,Q)\cap\mathcal N\bigr)=d^2-d.
	\]
	
	If, in addition, $W_0\in\mathcal U_H(P,Q)$ and $L_{X,W_0}$ is
	surjective, then, after shrinking $\mathcal N$ if necessary,
	$\mathcal U_H(P,Q)\cap\mathcal N$ is a real-analytic embedded
	submanifold with
	\[
	T_{W_0}\mathcal U_H(P,Q)=W_0\ker L_{X,W_0}
	\]
	and
	\begin{equation}
		\dim_{\mathbb R}
		\bigl(\mathcal U_H(P,Q)\cap\mathcal N\bigr)=d^2-d-2.
		\label{eq:generic-holevo-unitary-dimension}
	\end{equation}
	In particular, at $W_0=I$, the tangent spaces are $\mathfrak t_M$
	and $\ker L_X$, respectively, under the corresponding hypotheses.
\end{theorem}

\begin{proof}
  Use local coordinates $W=W_0e^K$.  Since $T_0$ has simple spectrum,
  analytic perturbation theory \cite[Chap.~II]{Kato1995} gives real-analytic
  local eigenvalue branches $\lambda_j(W)$ of $MW$ for $W$ near $W_0$,
  normalized by
  $\lambda_j(W_0)=\mu_j$.  By
  spectral continuity, we may choose disjoint disks about the $\mu_j$ in
  the open right half-plane and take the neighborhood so small that each
  disk contains exactly one eigenvalue of $MW$, counted with multiplicity.
  Thus reality of all $\lambda_j(W)$ implies that they are positive and distinct.
  Conversely, $MW\sim\Pd$ forces all of them to be real.  Hence, locally,
  attainability is equivalent to
  \[
  \operatorname{Im}\lambda_j(W)=0,
  \qquad 1\le j\le d.
  \]

  For $K\in\mathfrak{su}(d)$, standard first-order perturbation of a simple
  eigenvalue gives
  \[
  \left.\frac{d}{dt}\right|_{t=0}
  \lambda_j(W_0e^{tK})
  =
  \mu_j\Tr(E_jK).
  \]
  The functionals $K\mapsto\operatorname{Im}\Tr(E_jK)$ sum to zero
  because $\sum_j E_j=I$ and $\Tr K=0$. To see that this is their only
  independent relation, suppose that real numbers
  $c_1,\ldots,c_d$ satisfy
  \[
  \sum_{j=1}^dc_j\operatorname{Im}\Tr(E_jK)=0
  \qquad(K\in\mathfrak{su}(d)).
  \]
  Put $C:=\sum_jc_jE_j$.  Taking $K=iH$, where $H=H^*$ and $\Tr H=0$,
  shows that
  \[
  \Tr\bigl(\Herm(C)H\bigr)=0
  \]
  for every traceless Hermitian $H$.  Thus $\Herm(C)=aI$ for some real
  $a$.  It follows that $C=aI+S$ with $S^*=-S$, and therefore $C$ is
  normal with eigenvalues of the form $a+i\nu_j$.  On the other hand,
  $C$ is similar to $\operatorname{diag}(c_1,\ldots,c_d)$, so all of its
  eigenvalues are real.  Hence $S=0$, $C=aI$, and $c_1=\cdots=c_d=a$.
  This proves that the rank is exactly $d-1$.

  Choose $d-1$ branches whose imaginary-part differentials are independent.
  If their imaginary parts vanish, those $d-1$ eigenvalues are real and positive
  in a sufficiently small neighborhood of $W_0$.  Since their product with the
  remaining eigenvalue equals the fixed positive number $\det M$, the remaining
  eigenvalue is also real and positive.  Thus the attainable set is locally the zero
  set of these $d-1$ independent real equations.

  The real-analytic implicit-function theorem now gives the first manifold,
  whose dimension is
  \[
  (d^2-1)-(d-1)=d^2-d,
  \]
  and the preceding differential calculation gives its tangent space.

  If $W_0\in\mathcal U_H(P,Q)$, the local Holevo equation is
  \[
  H_{X,W_0}(W):=\Tr((W-W_0)X)=0.
  \]
  Its differential at $W_0$ in the tangent direction $W_0K$ is
  \[
  \mathrm DH_{X,W_0}(W_0)[W_0K]
  =\Tr(W_0KX)=\Tr(KXW_0)=L_{X,W_0}(K).
  \]
  Surjectivity of $L_{X,W_0}$ means that the two real components of
  $H_{X,W_0}$ have rank two on $\mathfrak t_{T_0}$.  A second application of the implicit-function
  theorem gives the asserted tangent space and subtracts two further real
  dimensions, proving \eqref{eq:generic-holevo-unitary-dimension}.
\end{proof}

At the identity, the two real conditions in the Holevo trace equation
are independent for every noncommuting pair. Working first with the
base as variable also covers repeated eigenvalues of $M$.

\begin{corollary}
\label{cor:generic-holevo-base-dimension}
Let $[P,Q]\ne0$. Near $R=I$, the Holevo bases form an embedded
real-analytic submanifold of $\Pd$ of dimension $d^2-2$. The
determinant-one bases form a local section for positive scaling and
have dimension $d^2-3$.
\end{corollary}

\begin{proof}
Write $C=P^{1/2}$ and $B=Q^{1/2}$, so that $M=C^{-1}B$ and $X=CB$.
Use the positive coordinates of
Theorem~\ref{thm:complete-classification-holevo-bases} and put
\[
 g(A)=\Tr(W_AX),\qquad W_A=\Pol(AM)^*.
\]
The map $A\mapsto R_A=C^{-1}A^2C^{-1}$ is a real-analytic
diffeomorphism. It sends $C$ to $I$, and the Holevo condition is
$g(A)=g(C)$. The key point is the surjectivity of the real
differential $J=\mathrm Dg(C)$ onto $\mathbb C$.

Fix $K=K^*$ and write
$S=\left.\frac{d}{dt}\right|_{0}\Pol((C+tK)M)$.
Since $CM=B>0$, differentiation of unitarity gives $S^*=-S$.
Differentiating the polar decomposition and taking the difference
with its adjoint gives
\[
 BS+SB=KC^{-1}B-BC^{-1}K,\qquad J(K)=-\Tr(SCB).
\]
For $K=C^2$, use an orthonormal eigenbasis $B=\operatorname{diag}(b_i)$.
The equation yields $S_{ij}=(b_j-b_i)C_{ij}/(b_i+b_j)$ and hence
\[
 J(C^2)=\sum_{i<j}\frac{(b_i-b_j)^2}{b_i+b_j}|C_{ij}|^2>0.
\]
The inequality follows from $[C,B]\ne0$. For an arbitrary Hermitian
$K$, skew-adjointness of $S$, cyclicity, and the same equation give
\[
 2i\operatorname{Im}J(K)
 =-\Tr\bigl((BS+SB)C\bigr)
 =-\Tr\bigl(K(C^{-1}BC-CBC^{-1})\bigr).
\]
Set $D=C^{-1}BC-CBC^{-1}$. Then $D^*=-D$, and $D=0$ would imply
$BC^2=C^2B$, hence $BC=CB$. Thus $D\ne0$. Taking the Hermitian
direction $K=iD$ gives
\[
2\operatorname{Im}J(iD)=\Tr(D^*D)>0.
\]
Together with the positive real value $J(C^2)$, this proves
surjectivity.

The real-analytic implicit-function theorem gives the first assertion.
Since $g(sA)=g(A)$ for $s>0$ and $R_{sA}=s^2R_A$, each scaling orbit
meets $\det R=1$ exactly once. The radial direction is tangent to the
Holevo locus and transverse to $\det R=1$, which gives the stated
dimension of the section.
\end{proof}

Passing from bases to factors removes the directions tangent to each
fiber. Proposition~\ref{prop:positive-cone-rank-and-connectedness}
determines the rank of the factor map at $R=I$ from the eigenvalue
multiplicities of $M$.

\begin{corollary}
\label{cor:generic-holevo-unitary-dimension}
Suppose $[P,Q]\ne0$. If $M=P^{-1/2}Q^{1/2}$ has simple spectrum,
the germ of $\mathcal U_H(P,Q)$ at $I$ is an embedded real-analytic
manifold of dimension $d^2-d-2$. In particular, this dimension is four
for qutrits with simple $M$. For every $d\ge3$, including pairs for
which $M$ has repeated eigenvalues, nonidentity Holevo factors approach
$I$. They can be realized by bases arbitrarily close to $I$.
\end{corollary}

\begin{proof}
	When $M$ has simple spectrum,
	Theorem~\ref{thm:local-holevo-unitary-geometry} gives the attainable
	manifold at $I$ with tangent space $\mathfrak t_M$ of dimension
	$d^2-d$. The differential of $A\mapsto W_A$ at $C=P^{1/2}$ takes
	values in $\mathfrak t_M$ and has rank $d^2-d$ by
	\eqref{eq:positive-cone-resolution-rank}. Hence it maps onto
	$\mathfrak t_M$. By the chain rule,
	\[
	J=L_X\circ\mathrm D_C W_A,
	\]
	where $J$ is the surjective differential in the preceding proof.
	Thus $L_X$ is surjective, and the Holevo dimension follows from
	Theorem~\ref{thm:local-holevo-unitary-geometry}.
	
	Now let $d\ge3$ and let $m_j$ be the multiplicities of the distinct
	eigenvalues of $M$. Since $P$ and $Q$ do not commute, $M$ cannot be
	scalar. By Theorem~\ref{thm:unitary-factor-classification-fixed-pair},
	the identity fiber has dimension
	\[
	\sum_jm_j^2\le(d-1)^2+1<d^2-2.
	\]
	By Corollary~\ref{cor:generic-holevo-base-dimension}, this fiber
	cannot contain a neighborhood of $R=I$ in the Holevo-base manifold.
	Choosing bases outside the fiber that converge to $I$ gives
	nonidentity Holevo factors converging to $I$ by continuity of the
	factor map.
\end{proof}

The simple-spectrum dimension applies on an open dense set of pairs.
Indeed, for fixed $P$, the matrix $M$ is similar to
$P^{-1/4}Q^{1/2}P^{-1/4}$, which ranges homeomorphically over $\Pd$
as $Q$ varies. Simple spectrum is open and dense in $\Pd$; the
noncommuting pairs are also open and dense for $d\ge2$.

\subsection{Simple spectrum for qutrits}
\label{sec:qutrit-constructive}

Fix $P,Q\in\mathbb P_3$ and consider the system \eqref{eq:fixed-gram-target}.
A solution $T=MW$ with simple spectrum has three left eigenlines.
Choosing one line leaves a $2\times2$ block to determine.

Choose a unit representative $v$ of a line $[v]\in\mathbb{CP}^2$
and an orthonormal frame $E$ for $v^\perp$. Write $F=[E\ v]$ and put
\begin{align}
 c&=(v^*Gv)^{1/2}, & z&=E^*Gv/c, & H&=E^*GE-zz^*,\label{eq:qutrit-schur}\\
 C&=E^*PE, & \eta&=h-v^*PGv/c, & K&=H^{1/2},\quad \kappa=\det K.
 \label{eq:qutrit-compressed}
\end{align}
The Schur complement of $v^*Gv$ in $F^*GF$ gives $H>0$.
For $q=(q_0,q_1,q_2,q_3)^T\in\mathbb R^4$, write
\[
 U(q)=\begin{pmatrix}
 q_0+\mathrm i q_1&q_2+\mathrm i q_3\\
 -q_2+\mathrm i q_3&q_0-\mathrm i q_1
 \end{pmatrix},\qquad B(q)=KU(q).
\]
As $q$ ranges over the unit sphere $S^3$, $U(q)$ ranges over $SU(2)$.
Define the real linear map $L:\mathbb R^4\to\mathbb R^3$ and the
vector $b$ by
\begin{equation}\label{eq:qutrit-linear-system}
 Lq=\begin{pmatrix}\Im\Tr B(q)\\ \Re\Tr(CB(q))\\ \Im\Tr(CB(q))\end{pmatrix},
 \qquad b=\begin{pmatrix}0\\ \Re\eta\\ \Im\eta\end{pmatrix}.
\end{equation}

\begin{theorem}\label{thm:qutrit-simple}
For each $[v]\in\mathbb{CP}^2$, take the solutions of $Lq=b$ on $S^3$.
For $t=\Tr B(q)$, retain those satisfying
\begin{equation}\label{eq:qutrit-simple-tests}
 t>2\sqrt\kappa,\qquad t\ne c+\kappa/c .
\end{equation}
Then the matrices
\begin{equation}\label{eq:qutrit-reconstruct}
 T=F\begin{pmatrix}B(q)&z\\0&c\end{pmatrix}F^*,\qquad W=M^{-1}T
\end{equation}
give precisely the factors $W\in\mathcal U_H(P,Q)$ for which $MW$
has three distinct eigenvalues. The set obtained for a fixed line
is independent of its representative and of the frame $E$.
\end{theorem}

\begin{proof}
Suppose $T$ satisfies \eqref{eq:fixed-gram-target} and has simple spectrum.
Fix a unit left eigenvector $v$, so $T^*v=c_0v$ for some $c_0>0$.
In the frame $F$, write $F^*TF=\left(\begin{smallmatrix}B&z_0\\0&c_0\end{smallmatrix}\right)$.
Comparing blocks in $F^*TT^*F=F^*GF$ gives
\[
 c_0^2=v^*Gv,\qquad c_0 z_0=E^*Gv,\qquad BB^*=H .
\]
Consequently $c_0=c$, $z_0=z$, and $K^{-1}B$ is unitary.
The eigenvalues of $B$ are positive, so $\det B=\det K=\kappa$
and $K^{-1}B\in SU(2)$. Write $B=B(q)$ with $q\in S^3$.
Since $Gv=cEz+c^2v$, the trace equation becomes
\[
 \Tr(PT)=\Tr(CB)+v^*PEz+c\,v^*Pv
        =\Tr(CB)+v^*PGv/c=h .
\]
Together with $\Im\Tr B=0$ this is $Lq=b$. The characteristic
polynomial of $B$ is $s^2-ts+\kappa$. Its two roots are positive and
distinct exactly when $t>2\sqrt\kappa$. Neither root equals $c$
exactly when $t\ne c+\kappa/c$, proving \eqref{eq:qutrit-simple-tests}.

Conversely, for a retained $q$, the same block identities give
$TT^*=G$ and $\Tr(PT)=h$. The two tests give three distinct positive
eigenvalues, hence $T\sim\mathbb P_3$. By
\eqref{eq:fixed-gram-target}, $W=M^{-1}T$ is a Holevo factor.
For the fixed line the construction has thus recovered exactly the
simple-spectrum solutions for which $T^*v=cv$. This description depends
only on $[v]$, proving the asserted independence.

For completeness, the sphere intersection in the construction has the
following explicit form. Let $q_*=L^\dagger b$, where $L^\dagger$ is the
Moore-Penrose inverse. The equation $Lq=b$ is consistent exactly when
$Lq_*=b$. In that case its solutions form $q_*+\ker L$, where
$q_*\perp\ker L$. Since $\dim\ker L\ge1$, this affine space meets $S^3$
exactly when $\|q_*\|\le1$, and the intersection consists of
\begin{equation}\label{eq:qutrit-sphere-solution}
 q=q_*+\sqrt{1-\|q_*\|^2}\,u,\qquad u\in\ker L,\quad \|u\|=1.
\end{equation}
Thus a line with $\operatorname{rank}L=3$ gives at most two candidates.
\end{proof}

\subsection{Repeated spectrum for qutrits}
\label{sec:qutrit-repeated-main}

For $P,Q\in\mathbb P_3$, the singular values of $M$ are fixed throughout the
locus $TT^*=G$. If $T$ has a repeated eigenvalue, these singular values
and the Holevo trace equation determine $T$ uniquely whenever it exists.
The endpoint comparison in
Corollary~\ref{cor:x-polar-monotone} gives
\begin{equation}\label{eq:qutrit-endpoint}
 \Tr(PG^{1/2})=\FM(P,Q)\le h,
 \qquad \FM(P,Q)=h\ \Longleftrightarrow\ [P,Q]=0.
\end{equation}

\begin{theorem}\label{thm:qutrit-repeated}
For every $P,Q\in\mathbb P_3$, at most one $W\in\mathcal U_H(P,Q)$ makes $MW$ have a repeated
eigenvalue. Let $0<r_1\le r_2\le r_3$ be the singular values of $M$.

If these singular values are not all distinct, a repeated-spectrum Holevo
factor exists exactly when $[P,Q]=0$, and it is $W=I$.

Suppose $r_1<r_2<r_3$. Choose $Z$ unitary so that
\[
 Z^*GZ=\operatorname{diag}(r_1^2,r_2^2,r_3^2),\qquad D_P=Z^*PZ,
 \qquad \lambda=r_2,\quad \mu=r_1r_3/r_2.
\]
If $\lambda=\mu$, there is no repeated-spectrum Holevo factor.
If $\lambda\ne\mu$, define
\begin{equation}\label{eq:qutrit-exception-data}
 a=\frac{\lambda+\mu}{r_1+r_3},\qquad b_0=\sqrt{1-a^2},\qquad
 \ell=a(r_1(D_P)_{11}+r_3(D_P)_{33})+r_2(D_P)_{22}.
\end{equation}
Then $0<a<1$. A repeated-spectrum factor exists exactly when
\begin{equation}\label{eq:qutrit-exception-test}
 (D_P)_{13}\ne0,\qquad h-\ell=b_0(r_3-r_1)|(D_P)_{13}|.
\end{equation}
When it exists, put $\xi=(D_P)_{13}/|(D_P)_{13}|$. Its formula is
\begin{equation}\label{eq:qutrit-exception-factor}
 T_{\mathrm{rep}}=Z
 \begin{pmatrix}
 r_1a&0&-r_1b_0\xi\\
 0&r_2&0\\
 r_3b_0\overline\xi&0&r_3a
 \end{pmatrix}Z^*,\qquad W_{\mathrm{rep}}=M^{-1}T_{\mathrm{rep}}.
\end{equation}
The eigenvalues of $T_{\mathrm{rep}}=MW_{\mathrm{rep}}$ are
$\lambda,\lambda,\mu$.
\end{theorem}
\begin{proof}
Let $T$ satisfy \eqref{eq:fixed-gram-target} and have repeated spectrum. If all three
eigenvalues equal $\lambda$, diagonalizability forces $T=\lambda I$.
Then $G=\lambda^2I$, $Q=\lambda^2P$, $M=\lambda I$, and $W=I$.

Otherwise $T=\lambda I+(\mu-\lambda)J$ for distinct positive
$\lambda,\mu$ and a rank-one idempotent $J$. To obtain its elementary
unitary normal form, choose a unit vector spanning $\operatorname{Ran} J$ and write
$J=uv^*$ with $v^*u=1$. Decompose $v=u+w$ with $w\perp u$ and choose
the second basis vector along $w$ when $w\ne0$. In this basis
\[
 J=\begin{pmatrix}1&s&0\\0&0&0\\0&0&0\end{pmatrix},\quad s\ge0,
 \qquad
 T=\begin{pmatrix}\mu&(\mu-\lambda)s&0\\0&\lambda&0\\0&0&\lambda\end{pmatrix}.
\]
If $s=0$, $T$ is positive Hermitian, so $T=G^{1/2}$.
If $s>0$, the characteristic polynomial of the squared singular values of
the leading block, evaluated at $\lambda^2$, equals
$-\lambda^2(\mu-\lambda)^2s^2<0$. Hence the two singular values of
that block lie strictly on opposite sides of $\lambda$. Consequently
the singular values of $T$ are distinct, $\lambda=r_2$, and determinants
give $\mu=r_1r_3/r_2$.

It follows that repeated singular values force $T=G^{1/2}$.
Its Holevo trace identity is precisely $\FM(P,Q)=h$, which by
\eqref{eq:qutrit-endpoint} is equivalent to $[P,Q]=0$. In that case $M=G^{1/2}$
and $W=I$, and its spectrum is indeed repeated.

Now assume the singular values are distinct. The positive Hermitian case
cannot have a repeated eigenvalue, so $s>0$. The third coordinate of the
displayed normal form is a reducing eigenline for both $T$ and $G$, with
eigenvalues $r_2$ and $r_2^2$. Therefore in the ordered $G$-eigenbasis, $T$
is $r_2$ in the middle coordinate and a $2\times2$ block $B$ in coordinates
$1,3$, with
\[
 BB^*=\operatorname{diag}(r_1^2,r_3^2),\qquad \det B=r_1r_3,
 \qquad \Tr B=\lambda+\mu.
\]
Write
\[
 B=\begin{pmatrix}r_1&0\\0&r_3\end{pmatrix}
   \begin{pmatrix}\alpha&\beta\\-\overline\beta&\overline\alpha\end{pmatrix},
 \qquad |\alpha|^2+|\beta|^2=1.
\]
The real trace and $r_1\ne r_3$ force $\alpha=a$ real, and then
$\beta=b_0e^{\mathrm i\theta}$. The identity
\[
 (r_1+r_3)-(r_2+r_1r_3/r_2)
 =\frac{(r_2-r_1)(r_3-r_2)}{r_2}>0
\]
proves $0<a<1$. If $\lambda=\mu$, diagonalizability would force $T=\lambda I$,
contrary to distinct singular values. Otherwise every matrix on this circle
is diagonalizable with positive spectrum $\lambda,\lambda,\mu$.

Write $T(\theta)$ for the matrix on this circle with
$\beta=b_0 e^{\mathrm i\theta}$. Its weighted trace is
\begin{equation}\label{eq:qutrit-phase-trace}
 \Tr(PT(\theta))=\ell+b_0\bigl(r_1(D_P)_{31}e^{\mathrm i\theta}
                         -r_3(D_P)_{13}e^{-\mathrm i\theta}\bigr).
\end{equation}
At a Holevo factor this trace equals $h$. Its imaginary part forces $(D_P)_{31}e^{\mathrm i\theta}$ to be real.
Since $a<1$ and $D_P>0$, \eqref{eq:qutrit-endpoint} gives
\[
 \ell<r_1(D_P)_{11}+r_2(D_P)_{22}+r_3(D_P)_{33}
      =\Tr(PG^{1/2})\le h.
\]
Thus $(D_P)_{13}\ne0$ and the only possible phase obeys
$(D_P)_{31}e^{\mathrm i\theta}=-|(D_P)_{13}|$.
Substitution gives exactly \eqref{eq:qutrit-exception-test} and
\eqref{eq:qutrit-exception-factor}. Conversely those formulas have the required
Gram matrix, trace, and positive diagonalizable spectrum. Changing the
phases of the columns of $Z$ leaves the reconstructed $T$ unchanged.
\end{proof}

\begin{corollary}
	\label{cor:qutrit-exception-closure}
	Let $P,Q\in\mathbb P_3$ be noncommuting. Then the Holevo factors
	$W$ for which $MW$ has simple spectrum are dense in
	$\mathcal U_H(P,Q)$.
\end{corollary}

\begin{proof}
	It suffices to consider the exceptional factor $W$ of
	Theorem~\ref{thm:qutrit-repeated}, when it exists. Put $T=MW$ and
	choose $A>0$ in the fiber of $W$, so that $B:=AT>0$. For $C>0$, write
	\[
	T_C=C^{-1}(CGC)^{1/2},\qquad f(C)=\Tr(PT_C).
	\]
The Holevo equation in these coordinates is $f(C)=h$.
It suffices to establish that the real differential of $f$ at $A$
has rank two.
	
	Consider the curve $C=A+tP$, and denote the derivative of $T_C$
	at $t=0$ by $\dot T$. Put $Y=T^{-1}\dot T$. Differentiating
	$T_CT_C^*=G$ gives
	\[
	\dot T T^*+T\dot T^*=0,
	\]
	so $Y+Y^*=0$. Differentiating $CT_C=T_C^*C$ and substituting
	$\dot T=TY$ gives
	\[
	BY+YB=-(PT-T^*P).
	\]
	Consequently,
	\[
	\Re\Tr(P\dot T)
	=\tfrac12\Tr\bigl((PT-T^*P)Y\bigr)
	=\Tr(BY^*Y).
	\]
	This is positive unless $Y=0$. If $Y=0$, then $PT=T^*P$.
	Thus $P^{1/2}TP^{-1/2}$ is Hermitian with positive eigenvalues,
	and hence $PT>0$. Since $PT=XW$ with $W$ unitary, we obtain
	\[
	h=\Tr(PT)=\|XW\|_1=\|X\|_1=\FU(P,Q),
	\]
	contrary to $h=\FH(P,Q)<\FU(P,Q)$.
	
	For a second direction, let $\theta_*$ be the phase of $T$ in
	\eqref{eq:qutrit-phase-trace}. The exceptional-factor formula gives
	$(D_P)_{31}e^{\mathrm i\theta_*}=-|(D_P)_{13}|\ne0$. Set
	\[
	D_s=Z\operatorname{diag}(e^{\mathrm i s},1,1)Z^*,
	\qquad C_s=D_sAD_s^*.
	\]
	Since $D_s$ commutes with $G$, we have
	$T_{C_s}=D_sTD_s^*=T(\theta_*+s)$. Differentiating
	\eqref{eq:qutrit-phase-trace} yields
	\[
	\left.\frac{d}{ds}\right|_{s=0}f(C_s)
	=-\mathrm i\,b_0(r_1+r_3)|(D_P)_{13}|.
	\]
	This nonzero purely imaginary derivative, together with the positive
	real part of the first derivative, proves that $df_A$ has rank two.
	
	By the implicit function theorem, $f^{-1}(h)$ is a seven-dimensional
	real-analytic manifold near $A$. On the other hand, the fiber of $W$
	is an embedded manifold of dimension $2^2+1^2=5$ by
	Theorem~\ref{thm:unitary-factor-classification-fixed-pair}.
	Hence there exist points $A_n\to A$ in $f^{-1}(h)$ outside this fiber.
	The corresponding Holevo factors $W_{A_n}$ converge to $W$ by
	continuity, and each $MW_{A_n}$ has simple spectrum by uniqueness
	of the exceptional factor.
\end{proof}

The following example has an exceptional factor distinct from $I$. Take
\[
 P^{1/2}=\operatorname{diag}(1,4,1),\qquad
 Q^{1/2}=\begin{pmatrix}6&2\\2&3\end{pmatrix}\oplus 1,\qquad
 W=\left(\frac15\begin{pmatrix}4&3\\-3&4\end{pmatrix}\right)\oplus 1.
\]
Here $[P,Q]\ne0$ and $\Tr(PMW)=\Tr(PM)=19$. The leading $2\times2$ block of $MW$ has
trace $9/2$ and determinant $7/2$, hence distinct eigenvalues $1$ and $7/2$.
Its direct sum with $1$ is diagonalizable with repeated eigenvalue $1$.
Thus \eqref{eq:fixed-gram-target} gives a Holevo factor $W$, and
Theorem~\ref{thm:qutrit-repeated} shows that it is the only factor for this
pair with repeated spectrum of $MW$. The same pair illustrates
Corollary~\ref{cor:generic-holevo-unitary-dimension}: the leading block
of $M$ has trace $27/4$ and determinant $7/2$, so its spectrum avoids
the last diagonal entry $1$. Thus $M$ has simple spectrum, and the corollary gives a four-dimensional manifold of factors near $I$.

\section{Further questions}
\label{sec:discussion}

Is the Holevo-factor locus connected for every noncommuting qutrit
pair? More generally, how many connected components can it have?
In dimensions $d\ge4$, which eigenvalue multiplicities can occur for
$P^{-1/2}Q^{1/2}W$ as $W$ varies over the Holevo factors of a fixed
pair, and how do the corresponding loci meet?

 \section*{Declaration of generative AI use}
During the preparation of this manuscript, the authors used OpenAI's ChatGPT to assist with language editing, manuscript organization, and the exploration and verification of mathematical arguments. All AI-assisted suggestions were critically reviewed by the authors. The authors independently verified the mathematical statements and proofs and take full responsibility for the accuracy, originality, and content of the manuscript.

\end{document}